\documentclass[10pt, a4paper]{scrartcl}
\usepackage{geometry}

\usepackage[T1]{fontenc}
\usepackage[utf8]{inputenc}
\usepackage[english]{babel}
\usepackage{amsmath,amssymb,amstext}
\usepackage{esint}
\usepackage{nicefrac}
\usepackage{hyperref}
\usepackage{pict2e}
\usepackage{graphicx}
\usepackage{mathdots}
\usepackage{xcolor}
\usepackage[linesnumbered,ruled,lined]{algorithm2e}
\usepackage{dsfont}
\usepackage{enumitem}
\usepackage{bm}
\usepackage{mathtools}

\usepackage{authblk}

\usepackage{amsthm}

\newtheorem{z}{}[section]
\newtheorem{thm}[z]{Theorem}
\newtheorem{lem}[z]{Lemma}
\newtheorem{alg}[z]{Algorithm}
\newtheorem{examp}[z]{Example}
\newtheorem{defin}[z]{Definition}

\newcommand{\N}{\mathbb{N}}

\newcommand{\R}{\mathbb{R}}

\newcommand{\E}{\mathbb{E}}
\newcommand{\id}{\mathrm{d}}
\DeclareMathOperator{\eps}{\varepsilon}
\DeclareMathOperator{\var}{var}

\newcommand{\spc}{\mathcal{C}}

\newcommand\numberthis{\addtocounter{equation}{1}\tag{\theequation}}
\newcommand{\Dto}{\overset{\mathcal{D}}{\longrightarrow}}
\newcommand{\Pto}{\overset{\mathrm{P}}{\longrightarrow}}

\title{\LARGE Convergence of random-weight sequential Monte Carlo methods}
\author[1]{Paul~B.~Rohrbach}
\author[1,2]{Robert~L.~Jack}
\affil[1]{Department of Applied Mathematics and Theoretical Physics, University of Cambridge, Wilberforce Road, Cambridge CB3 0WA, United Kingdom}
\affil[2]{Yusuf Hamied Department of Chemistry, University of Cambridge, Lensfield Road, Cambridge CB2 1EW, United Kingdom}
\date{\today}

\begin{document}
\maketitle 

\begin{abstract}
We investigate the properties of a sequential Monte Carlo method where the particle weight that appears in the algorithm is estimated by a positive, unbiased estimator.
We present broadly-applicable convergence results, including a central limit theorem, and we discuss their relevance for applications in statistical physics.
Using these results, we show that the resampling step reduces the impact of the randomness of the weights on the asymptotic variance of the estimator.
In addition, we explore the limits of convergence of the sequential Monte Carlo method, with a focus on almost sure convergence.
We construct an example algorithm where we can prove convergence in probability, but which does not converge almost surely, even in the non-random-weight case.
\end{abstract}


\section{Introduction}

Sequential Monte Carlo (SMC) methods are algorithms that approximate a sequence of probability distributions by a sequence of interacting particle systems, using importance sampling and resampling.
These computational tools are very widely used in many disciplines including statistics, physics, and finance \cite{del2006sequential,ionides2006inference,hsu2011review,jasra2011sequential}. 
In signal processing for example, these methods are well known for their usage in prediction, filtering, and smoothing on hidden Markov models \cite{cappe2005inference,doucet2009tutorial}, where they are often also referred to as particle filters.
Methods related to SMC have also been exploited in problems from statistical physics~\cite{iba2001population,grassberger2002go,angeli2019rare}.
In particular, recent work on coarse-grained modelling of soft matter systems proposed an SMC method where unbiased estimates of weights are used for correction of the distribution of the SMC particles~\cite{rohrbach2022multilevel}.
In the following, we refer to this method as the reverse-mapping SMC (RMSMC) method.
The random weights are essential in that setting, because the weights are evaluated as free-energy differences, for which unbiased estimates are available, but exact results are not.
A numerical characterisation of the performance of that method was given in~\cite{rohrbach2022multilevel}.
A central motivation of this work is to provide a rigorous basis for its convergence.

While results for convergence of SMC methods are well-established~\cite{del1996nonlinear,del1999central,kunsch2005recursive}, dealing with random weights requires some extra analysis (a discussion of previous work in this direction is given below).  
Our results apply within a generic algorithmic setting for random-weight SMC, inspired by~\cite{chopin2004central}, complementing previous analysis that focussed on the specific case of particle filters~\cite{fearnhead2008particle,fearnhead2010random}.
We extend previous work in several ways.
First, we prove a central limit theorem (CLT) in our general setting, following the method of~\cite{cappe2005inference}.
This theorem is sufficiently general to provide a rigorous foundation for the RMSMC method.
Second, we analyse the impact that the randomness of the weights has on the asymptotic variance of the SMC estimator.
As might be expected, randomness in the weights increases the asymptotic variance, but we show that this can be ameliorated by suitable resampling, within SMC.
We provide numerical results to illustrate this effect.
Thirdly, we discuss almost-sure convergence of SMC methods, including algorithms within the setting of~\cite{chopin2004central}, and extensions to random weights.
We provide an example (with non-random weights) to show that convergence in probability of SMC is not sufficient to ensure almost-sure convergence, consistent with previous results in less general settings~\cite{crisan2000convergence,crisan2002survey,hu2008basic}, but contrary to the discussion of~\cite{chopin2004central}.
Finally, we discuss the implications of this last result for proofs of CLTs within SMC (in particular, we explain that a direct generalization of the proof of~\cite{chopin2004central} is not sufficient to derive the CLTs quoted here).

\subsection{Connections to previous work}

For SMC with non-random weights, broadly-applicable convergence results have been established, including convergence in probability, and CLTs \cite{chopin2004central,cappe2005inference,douc2008limit,chan2013general}.
Under more constrained conditions, one can also prove the almost sure convergence of SMC estimators \cite{crisan2000convergence,crisan2002survey,hu2008basic}.
We refer to \cite{chopin2020introduction} for an introduction and to \cite{del2004feynman} for a detailed analysis of the theoretical properties of SMC.
In general, the challenges of such analyses are that the particles of the simulated system are not independent, so that standard MC convergence results do not apply.

As well as these results, the modern theory of particle filters includes finer convergence results than the CLT, for example \cite{whiteley2013stability,douc2014long}, which hold for SMC methods in the limit of long sequences.
However, the RMSMC method does not use long sequences (the only relevant limit is the number of particles).
The variance that appears in the CLT is the natural error estimate in such applications, hence our focus on CLTs in this work.

Previous work has also developed generalised versions of the SMC method that allow the weights to be estimated stochastically, including SMC$^2$ \cite{chopin2013smc2} for the analysis of state space models, particle filters for continuous time processes \cite{fearnhead2008particle,fearnhead2010random}, and the nested sequential Monte Carlo algorithm \cite{naesseth2015nested}.
More generally, replacing an intractable factor in a Monte Carlo (MC) scheme by an unbiased estimate has been proven to be useful in extending their applicability.
For example, the pseudo marginal Metropolis Hastings algorithm \cite{andrieu2009pseudo} makes it possible to compute estimates with respect to distributions where the density cannot be computed exactly.
Such methods are relevant in physics~\cite{kobayashi2019correction,kobayashi2021critical,rohrbach2022multilevel}, because 
the proposal and target distribution of SMC are linked by an estimated free energy difference, which has to be estimated by a stochastic annealing procedure between the two distributions.

For analysis  of the convergence of such random-weight methods, a general proposal of \cite{fearnhead2008particle} was to include the random weights within the particles' states, and then to apply the method of~\cite{chopin2004central}.
However, the technical details of such a proof would require that almost-sure convergence of the algorithm is established, before proving the CLT.
We explain below that this is not possible in general.
To avoid any confusion about such technical details, we provide (for completeness) a separate proof of the CLT, for random and non-random weights.
The proof is obtained as a straightforward translation of the approach of~\cite{cappe2005inference}, from its original formulation in hidden Markov models, to our algorithmic setting.
As such, the main messages of this work are the general setting for random-weight methods, which includes the RMSMC method; the interpretation of the asymptotic variance and its implications for the performance of practical methods; the differences in conditions for almost-sure convergence and convergence in probability; and the comparison of different proof strategies for CLTs.

This paper is organised as follows.
In Section \ref{sec:setup}, we give a brief overview of the SMC algorithm and its random-weight generalisation, and we outline its relationship to the physics application of~\cite{rohrbach2022multilevel}.
Section \ref{sec:alternatives} presents our convergence results, including a central limit theorem.
We give a simple numerical example, showing how resampling can be optimised to reduce the detrimental effects of randomness in the weights.
In Section \ref{sec:failure}, we construct an example where the SMC algorithm does not converge almost surely and discuss the impact this has on the results of \cite{chopin2004central}.

\section{Structure of the SMC method} \label{sec:setup}

This section defines the SMC algorithms that we consider.
We describe the non-random-weight version of the algorithm following the approach of \cite{chopin2004central} before generalising it to allow for random weights.

\subsection{Non-random weights}

The target of the SMC algorithm is a sequence of probability distributions $\pi_t$, $t=0, \dots, T$, on measurable spaces $(\spc_t, \Omega_t)$.
At each step $t$, SMC approximates the distribution $\pi_t$ by a weighted set of $H \in \N$ particles.
We assume here that $H$ is fixed (independent of $t$) although this assumption is easily relaxed at the expense of a more cumbersome notation.
The position of the $j$th particle is $C_t^{j,H} \in \spc_t$ and its weight is $W_t^{j,H} \in \R_{>0}$: these are random variables.  (Uppercase letters are used to denote random variables, throughout the text.)
The particles approximate the distribution $\pi_t$ in the sense that for a suitable measurable quantity of interest $\phi: \spc_t \to \R$, the SMC estimator
\begin{equation} \label{eqn:smcestimator}
    \mathbb{E}_{\pi_t} [\phi(C_t)]  \approx \sum_{j=1}^H \frac{W_t^{j,H}}{\sum_{k=1}^H W_t^{k,H} } \phi(C^{j,H}_t) 
\end{equation}
approximates the expectation of $\phi$ with respect to $\pi_t$.
(We write $\mathbb{E}_{\pi_t} [\phi(C_t)]$ for expectations of $C_t\sim \pi_t$.)

To generate the particles at step $t$, we mutate the particles from the previous step $t-1$ by a Markov transition kernel
$
    k_t: \spc_{t-1} \times \Omega_t \to [0, 1]
$
that maps $\spc_{t-1}$ to the space $\mathcal{P}(\spc_t)$ of probability distributions on $\spc_{t}$.
We denote the distribution of the mutated particles by
\begin{equation}
    \tilde \pi_t(\id c_t) = (\pi_{t-1} k_t)(\id c_t) = \int_{\mathcal{C}_{t-1}} \pi_{t-1}(\id c_{t-1}) k_t(c_{t-1}, \id c_t) .
\end{equation}
Generally, this distribution will differ from the target distribution $\pi_t$ at step $t$.
We assume that $\pi_t$ is absolutely continuous with respect to $\tilde \pi_t$ and compensate for this difference by reweighting the particles with the Radon-Nikodym derivative
\begin{equation} \label{eqn:defweight}
    v_t(c_t) := \frac{\id \pi_t}{\id \tilde \pi_t}(c_t) ,
\end{equation}
for $c_t \in \mathcal{C}_t$.
We assume further that there exists a strictly positive version of the weights $v_t > 0$.
In applications, the Radon-Nikodym derivative in \eqref{eqn:defweight} is typically known only  up to constant factors.
For the algorithm, it is only necessary to evaluate a function $w_t$ that satisfies
\begin{equation} \label{eqn:defuweight}
    w_t(C_t) = z_t v_t(C_t),
\end{equation}
where $z_t$ is an (unknown) constant.
This constant is not required when implementing the SMC algorithm, but it will be convenient to work with the normalised weights $v_t$ in the proofs.

The SMC algorithm begins by initialising the particles from some instrumental distribution $\pi_0$.
Then, the SMC algorithm iteratively applies the steps described above for a fixed population size $H$, combined with a resampling step that replaces the set of weighted particles by a set of unweighted ones.
\begin{alg}[Sequential Monte Carlo] \label{alg:smc}
    \normalfont
    Start by initialising an unweighted set of $H$ particles $(\hat C_0^{j, H})_{j \leq H}$ whose elements are independently and identically distributed (i.i.d.)~samples from $\pi_0$.
    Iterate the following steps for $t=1, 2, \dots, T$:
    \begin{enumerate}
    \item Mutation:
        Given an unweighted set $(\hat C_{t-1}^{j, H})_{j \leq H}$ which approximates $\pi_{t-1}$, mutate each particle by drawing from $k_t$ to obtain
    	\begin{equation}
    		C_t^{j, H} \sim k_t(\hat C_{t-1}^{j, H}, \cdot).
    	\end{equation}
    	These new particles (with unit weight) approximate the intermediate distribution $\tilde \pi_t$.
    \item Correction:
    	Assign to each mutated particle a weight 
    	\begin{equation} \label{eqn:defofweight}
    		W_t^{j, H} = w_t(C_t^{j,H})
    	\end{equation}
    	such that the weighted particles $(C_t^{j, H}, W_t^{j, H})_{j \leq H}$ approximate $\pi_t$.
    	In particular, one may estimate $\E_{\pi_t}[\phi(C_t)]$ by \eqref{eqn:smcestimator}.
    \item Selection:
    	Resample the weighted set of particles to obtain an unweighted set
    	\begin{equation}
    		(C_t^{j, H}, W_t^{j, H})_{j \leq H} \longrightarrow (\hat C_t^{j, H})_{j \leq H}
    	\end{equation}
    	using multinomial resampling \cite{douc2005comparison}.
    	The resulting unweighted set $(\hat C^{j,H}_t)_{j \leq H}$ approximates $\pi_t$.
    \end{enumerate}
\end{alg}
We call this the non-random-weight SMC algorithm, since the weights $W_t^{j,H}$ are related to the particles $C_t^{j, H}$ by a deterministic function $w_t$.  
This algorithm is very similar to that of~\cite{chopin2004central}, although the initialisation step is formulated in a slightly different way, which facilitates the generalisation to random weights.
\begin{examp}[Particle filter] \label{ex:ParticleFilter}
    The canonical example of an SMC method is the particle filter which was first introduced in 1993 by \cite{gordon1993novel} to solve the filtering problem for hidden Markov models.
    We show that the particle filter is a variant of Algorithm \ref{alg:smc}.
    This example follows \cite{doucet2009tutorial} in presentation, see that work for further details.
    For simplicity, we assume here that all distributions and kernels admit densities with respect to the Lebesgue measure.
    Let $(X_t)_{0 \leq t \leq T}$, $X_t \in \R^n,$ be a discrete-time Markov process with initial distribution $X_0 \sim \mu_0(x_0)$ given by a probability density $\mu_0$ and
    \begin{equation}
        X_t \mid (X_{t-1} = x_{t-1}) \sim f_t(x_{t-1}, \cdot),
    \end{equation}
    with Markov transition kernel $f_t$.
    For a hidden Markov model, this $X_t$ (the hidden process) is observed indirectly, via an observation process $(Y_t)_{1 \leq t \leq T}$, $Y_t \in \R^{m}$, which is conditionally independent given $(X_t)_t$ with marginal probability
    \begin{equation}
        Y_t \mid (X_t = x_t) \sim g_t(y_t \mid x_t), \;\; t=1, \dots, T.
    \end{equation}
    Having observed the process $Y_{1:T} = y_{1:T}$, where $z_{i:j} = (z_i, z_{i+1}, \dots, z_j)$, our goal is to investigate the posterior distribution of the hidden process given the observations
    \begin{equation}
        \pi_T(x_{0:T}) \propto \mu_0(x_0) \prod_{t=1}^{T} f_t(x_{t-1}, x_t) g_t(y_t \mid x_t), \;\; x_{0:T} \in \R^{(T+1) \times m}
    \end{equation}
    and, in particular, the filtering distribution $\pi_T(x_T) = \int \pi_t(x_{0:T}) \id x_{0:T-1}$.
    
    We want to use an SMC algorithm to generate samples from the full posterior distribution $\pi_T(x_{0:T})$.
    For this, we need to define a mutation step and corresponding kernel $k_t$ for which the evaluation of the particle weight $w_t$ is tractable.
    Here, the dimension of the particles $C_t = x_{0:t} \in \R^{(t+1) \times n}$ increases with each step $t$.
    An intuitive way of proposing a new particle $x_{0:t}$ given $x_{0:t-1}$ is to predict its next entry $x_{t}$ via a Markov transition kernel
    \begin{equation}
        x_{t} \sim k_{t}(x_{0:t-1}, \cdot).
    \end{equation}
    Then, the intermediate distribution is given by
    \begin{equation}
        \tilde \pi_{t}(x_{0:t}) \propto \pi_{t-1}(x_{0:t-1}) k_{t-1}(x_{0:t-1}, x_t)
    \end{equation}
    and we can compute the particle weight up to a constant
    \begin{equation}
        w_t(x_{0:t}) = \frac{f_{t+1}(x_{t-1}, x_t) g_t(y_t \mid x_t)}{k_t(x_{0:t-1}, x_t)}.
    \end{equation}
    This special case of Algorithm \ref{alg:smc} is an implementation of the particle filter algorithm \cite{doucet2009tutorial}.
\end{examp}

\subsection{Random weights}

To implement Algorithm \ref{alg:smc}, we require that the following operations are tractable: (i) sampling from $\pi_0$; (ii) mutation of samples via the kernels $(k_1,k_2,\dots)$; (iii) evaluation of the Radon-Nikodym derivatives $w_t$ up to a constant.
We can relax the last requirement (iii) by replacing the weight $w_t$ with a positive, unbiased estimate \cite{fearnhead2008particle,fearnhead2010random,chopin2013smc2}.
For this, we define a generalised Markov transition kernel $\bar k_t: \spc_{t-1} \times \Omega_t \otimes \mathcal{B}(\R_{>0}) \to [0,1]$,  where $\mathcal{B}(\R_{>0})$ is the Borel $\sigma$-algebra on the positive real numbers, that maps $\spc_t$ to the space $\mathcal{P}(\spc_t \times \R_{>0})$ of probability distributions of particles and weights. 
The generalised kernel simultaneously mutates a particle and provides an estimate of the weight $w_t$: for $C_{t-1} \sim \pi_{t-1}$ then
\begin{equation}
    (C_t, W_t) \sim \bar k_t(C_{t-1}, \cdot)
\end{equation}
generates a positive weight $W_t > 0$ such that
\begin{equation}
    \E[W_t \phi(C_t)] = \E[w_t(C_t) \phi(C_t)] = z_t \E_{\pi_t}[\phi]
\end{equation}
for any $\phi \in L^1(\spc_t, \pi_t)$.
In particular, $W_t$ is an unbiased estimate of the (non-random) weight, that is
\begin{equation} \label{eqn:WeightConditionalExp}
    \E[W_t | C_t] = w_t(C_t).
\end{equation}

The random-weight SMC algorithm replaces the weight in Algorithm \ref{alg:smc} by a random estimate, as follows.
\begin{alg}[Random-weight sequential Monte Carlo] \label{alg:rwsmc}
    \normalfont
    Start by initialising an unweighted set of $H$ particles $(\hat C_0^{j, H})_{j \leq H}$ whose elements are i.i.d.~samples from $\pi_0$.
    Iterate the following steps for $t=1, 2, \dots, T$:
    \begin{enumerate}
    \item Mutation and correction:
        Given an unweighted set $(\hat C^{j,H}_{t-1})_{j \leq H}$ which approximates $\pi_{t-1}$, mutate and estimate the weight of each particle by drawing from $\bar k_t$ to obtain
        \begin{equation}
    		(C_t^{j, H}, W_t^{j, H}) \sim \bar k_t(\hat C_{t-1}^{j, H}, \cdot)
    	\end{equation}
    	such that the weighted particles $(C_t^{j,H}, W_t^{j,H})$ approximate $\pi_t$.
    \item Selection:
    	Resample the weighted set of particles to obtain an unweighted set
    	\begin{equation}
    		(C_t^{j, H}, W_t^{j, H})_{j \leq H} \longrightarrow (\hat C_t^{j, H})_{j \leq H}
    	\end{equation}
    	using multinomial resampling (which is well defined as the weights are assumed to be positive).
    	The resulting unweighted set $(\hat C_t^{j, H})_{j \leq H}$ approximates $\pi_t$.
    \end{enumerate}
\end{alg}
This algorithm reduces to the non-random Algorithm~\ref{alg:smc} when $W_t = w_t(C_t)$. 

An example of a random-weight algorithm can be found in \cite{fearnhead2008particle} where a random-weight extension of the particle filter of Example \ref{ex:ParticleFilter} is considered.  
The algorithm that originally motivated this work comes from statistical physics~\cite{rohrbach2022multilevel}: in that case, each distribution $\pi_t$ is a grand-canonical Boltzmann-Gibbs distribution, as follows.

\newcommand*{\BPS}{X}
\newcommand*{\SPS}{x}

\begin{examp}[Hard sphere systems] \label{ex:HS}
    Consider a system of hard spheres in a bounded domain $\Lambda \subset \mathbb{R}^3$.  
    There are $N$ large spheres (radius $r_{\mathrm b}$) and $n$ smaller ones (radius $r_{\mathrm s} < r_{\mathrm b}$).
    The number of large spheres fluctuates between $0 \leq N \leq N_{\mathrm{max}}$, their configuration is denoted by
    \begin{equation}
        \BPS \in \mathcal{D} := \bigcup_{N=0}^{N_{\mathrm{max}}} \Lambda^N.
    \end{equation}
    Similarly, the $n$ small spheres have configuration $\SPS \in \mathcal{D}$ and the full system is denoted by $C = (\BPS, \SPS)$.
    The spheres are hard, so none of them can overlap.
    We take $N_{\mathrm{max}}$ larger than the number of small spheres in an optimal packing in $\Lambda$, so that $\mathcal{D}$ covers all possible configurations.
    
    The final (target) distribution is the grand-canonical binary hard sphere mixture \cite{Dijkstra1999-pre} with density
    \begin{equation}
        \pi_T(C) \propto \frac{\lambda_{\mathrm b}^N \lambda_{\mathrm s}^n}{N! n!} e^{- U_T(C)},
    \end{equation}
    where $\lambda_{\mathrm b}, \lambda_{\mathrm s} \in (0,\infty)$ are parameters; the energy $U_T(C) = \infty$ if any spheres overlap, and $U_T(C) = 0$ otherwise.  
    If the spheres were non-interacting ($U_T \equiv 0$), their numbers $N,n$ would follow independent (truncated) Poisson distributions.
    
    For small size ratios $q = r_{\mathrm s} / r_{\mathrm b} < 1$, the scale separation between large and small spheres makes it is very challenging to sample from $\pi_T$.  
    Physically, our primary interest is in the marginal distribution of the big spheres, $\pi_T(\BPS)  =  \int_{\mathcal D} \pi_T(\BPS, \SPS) \, \id \SPS$, which means that
    \begin{equation}
        \pi_T(\BPS) 
            \propto \frac{\lambda_{\mathrm b}^N}{N!} \Xi_{\lambda_\mathrm{s}}(\BPS)
    \end{equation}
    where 
    \begin{equation}
    \Xi_{\lambda}(\BPS) =  \int_{\mathcal D} \frac{\lambda^n}{n!} e^{- U_T(\BPS, \SPS)} \, \id \SPS
            \label{eqn:ExHSPartFunc}
    \end{equation}
    is called the partition function of the small spheres.
    This integral is intractable, except for the trivial regime of very small $\lambda_{\mathrm s}$.
    However, a positive, unbiased estimate of $\Xi_{\lambda}(\BPS)$  is available through Annealed Importance Sampling (AIS) \cite{jarzynski1997nonequilibrium,crooks2000path,neal2001annealed}, by simulating a process that slowly introduces the small spheres to the system, starting from a system with $n=0$.  This reintroduction of fine details to a physical system is known as reverse mapping \cite{Spyriouni2007,kobayashi2019correction}.
    The AIS procedure enables the development of an RMSMC method that efficiently samples $\pi_T$.
    
    The method starts from an approximate (coarse-grained, CG) model of the large spheres alone
    \begin{equation}
        \pi_0(\BPS) \propto \frac{\lambda_{\mathrm b}^N}{N!} e^{- U_0(\BPS)}
    \end{equation}
    which is designed to mimic the marginal distribution of interest $\pi_0(\BPS) \approx \pi_T(\BPS)$.
    Such models are available in this case \cite{RED,Ashton2011depletion,kobayashi2021critical}, and they do not suffer from
separation of length scales, so samples $\BPS \sim \pi_0$ are easily generated by MCMC.
    Estimates with respect to $\pi_T$ might then be computed by importance sampling with weight $w(\BPS) = \pi_T(\BPS) / \pi_0(\BPS)$
    but this $w(\BPS)$ is intractable, as discussed above.  Hence we use AIS to estimate these weights.
    However, this process is very expensive and it is efficient to split the annealing into multiple steps, within the framework of random-weight SMC.
    
    The method starts with the CG model which does not contain any small spheres, its density is $\pi_0(C) = \pi_0(\BPS) \delta_{n=0}(\SPS)$.
    We construct a sequence of intermediate distributions $\pi_t(C)$, $t=1, \dots, T$,  that include small spheres in increasingly large regions around the big spheres $\BPS$, until small spheres are inserted everywhere in $\Lambda$ at $t=T$, giving the target distribution.
    The information provided by the partially inserted small spheres $\SPS$ is used to design $\pi_t(C)$ such that their marginals $\pi_t(\BPS)$ become increasingly accurate approximations of the target marginal $\pi_T(\BPS)$, details are given in \cite{rohrbach2022multilevel}.
    
    To apply Algorithm \ref{alg:rwsmc} to this model hierarchy, 
    the SMC algorithm is initialised with samples from the CG model
    \begin{equation}
        \hat C_0^{j,H} \sim \pi_0(C), \;\;\; j=1, \dots, H.
    \end{equation}
    Given configurations $\hat C_{t-1}^{j,H}$ at step $t-1$, AIS is used to insert small spheres in the regions that are additionally populated in $\pi_t(C)$ (compared to $\pi_{t-1}(C)$), while keeping the big spheres $\BPS_{t-1}^{j,H}$ fixed.
    This results in a new configuration of small spheres $\SPS_t^{j,H}$, and an estimated ratio of small sphere partition functions which is used to compute the appropriate importance weight $W_t^{j,H}$.
    We summarise this process by a Markov kernel $\bar k_t$:
    \begin{equation}
        (\SPS_t^{j,H}, W_t^{j,H}) \sim \bar k_t(\hat C^{j,H}_{t-1}, \cdot).
    \end{equation}
    The updated population of weighted configurations is $(C_t^{j, H}, W_t^{j,H})$ with $C_t^{j,H} = (\BPS_{t-1}^{j,H}, \SPS_t^{j,H})$.
    
    In a second step, we resample the weighted set of configurations to get $\hat C_t^{j, H}$.
    Finally at $t=T$, this process yields a set of weighted configurations $(C_T^{j, H}, W_T^{j, H})$, $j=1,\dots,H$, that approximates the target distribution
    \begin{equation}
        \sum_{j=1}^{H} \frac{W^{j,H}_T}{\sum_{k=1}^{H} W_T^{k,H}} \phi(\BPS^{j,H}_T) \approx E_{\pi_T}[\phi]
    \end{equation}
    for any measurable function of the big spheres $\phi: \mathcal{D} \to \R$.
    
    This application of Algorithm \ref{alg:rwsmc} to the sequence of distributions $\pi_t(C)$ is an example of the reverse-mapping SMC method; a detailed discussion of this example with one intermediate level ($T=2$) can be found in~\cite{rohrbach2022multilevel}.
\end{examp}

\subsection{Notation conventions}

To simplify the presentation of the following sections, we define the following notations.
We write expectations of a function $\psi$ of the random variables $(C_t, W_t) \sim \bar k_t(c_{t-1}, \cdot)$ for a fixed $c_{t-1} \in \spc_{t-1}$ as
\begin{equation}
    \E_{\bar k_t(c_{t-1}, \cdot)}[\psi(C_t, W_t)]
\end{equation}
and the corresponding expectation with respect to $C_{t-1} \sim \pi_{t-1}$ as
\begin{equation}
    \E_{\tilde \pi_t}[\psi(C_t, W_t)] := \E_{\pi_{t-1}}\left[\E_{\bar k_t(C_{t-1}, \cdot)}[\psi(C_t, W_t)]\right],
\end{equation}
by analogy with expectations with respect to $\tilde \pi_t$ in the non-random-weight case.
In addition, we define the normalised weight random variable $V_t = W_t/z_t$ with the normalising constant from \eqref{eqn:defuweight}.
When it is clear from context, $c_t$ denotes an element from $\mathcal{C}_t$ and we use the notation $c_t \mapsto f(c_t)$ as a shorthand definition for the function $f: \mathcal{C}_t \to \R$.

In a numerical simulation, one would usually fix some $H$ and run the algorithm described above.
To analyse the convergence of the method as $H \to \infty$, we imagine a sequence of such computations, with increasing $H$.
The results of all such computations are collected into triangular arrays of random variables
\begin{equation}
	(\mathbf C_t, \mathbf W_t) 
	    = \left\{ \left(C_t^{j,H}, W_t^{j,H}\right)_{1 \leq j \leq H} \right\}_{H \in \N}
    \;\;\; \text{and} \;\;\;
    \hat{\mathbf C}_t = \left\{ \left( \hat C_t^{j,H} \right)_{1 \leq j \leq H} \right\}_{H \in \N}.
\end{equation}
In the following, we will use boldface letters for such triangular arrays.
To specify the relation of random variables for different rows $H$, we make the natural assumption that random variables on different rows of the triangular array are independent of each other.
This means that all elements of the initial array $\hat{\mathbf C}_0$ are sampled independently and that the steps of the algorithm are computed independently on each row.
This is consistent with the way one would compute an SMC estimate in practice, where it is necessary to rerun the algorithm with a new, bigger initial set of particles to increase the number of particles on level $t>1$.

\section{Convergence results}\label{sec:alternatives}
The resampling step of the SMC algorithm introduces inter-particle dependence in the system, which means that standard Monte Carlo convergence arguments are not applicable.
However, there are established approaches for proving convergence results for SMC methods~\cite{chopin2004central,cappe2005inference,douc2008limit}, which make use of the sequential nature of the algorithm.
When conditioning on the system before resampling, the resampled particles are conditionally independent.
This allows  convergence results to be formulated iteratively, using the conditional independence and the convergence from the previous step to establish convergence for the current step.

In this section, we apply this idea to the random-weight SMC Algorithm \ref{alg:rwsmc}.
In a first step, we show the convergence in probability and a central limit theorem under minor prerequisites for the weights and quantities of interest (see Theorems~\ref{thm:convinprob} and~\ref{thm:clt}).
Using this CLT, we analyse the effect that the randomness of the weights has on the SMC algorithm.
In a second step, we also prove almost sure convergence of the algorithm (Theorem~\ref{thm:asconvergence}).
In the context of applications, the important properties of the SMC method are (i) that it provides a consistent estimate in the large-number-of-particles limit $H \to \infty$ (as shown by its convergence in probability) and (ii) that we understand the magnitude of random fluctuations of an estimate at a fixed number of particles $H$ (which, at least asymptotically, can be obtained by the CLT).
We are not aware of practical applications where a.s.\ convergence is explicitly required -- the motivation for this part of our analysis is more theoretical, because of its implications for the proofs of the CLT in~\cite{chopin2004central}.
An important feature in this regard is that the proof of a.s.~convergence (Theorem~\ref{thm:asconvergence}) requires much more restrictive prerequisites than Theorems~\ref{thm:convinprob} and~\ref{thm:clt} -- this is in line with other previous results~\cite{crisan2000convergence,crisan2002survey}.

All theorems presented in this section are of course also valid results for the special case of the non-random-weight SMC Algorithm \ref{alg:smc} by setting $W_t = w_t(C_t)$.
The proofs are collected in Appendix \ref{sec:proofs}.

\subsection{Central limit theorem}\label{subsec:CLT}
To prove a weak law of large numbers and a central limit theorem, we apply the results from \cite{cappe2005inference} for hidden Markov models and \cite{douc2008limit} for general sequences of weighted samples to our Algorithm \ref{alg:rwsmc}.
The proofs presented there are for SMC with non-random weights.
However, the main ingredients of the proofs, which are the sequential structure of the algorithm and the conditional independence given the previous step, are not affected by the randomisation of the weights, so the application of the techniques of \cite{cappe2005inference,douc2008limit} to Algorithm~\ref{alg:rwsmc} is quite straightforward.

The applicability of this methodology confirms the observation of~\cite{fearnhead2008particle} that the generalisation from non-random- to random-weight SMC methods should be simple to accomplish, since it can also be achieved by augmenting the state variables of the particles.
(That work proved a CLT for a random-weight particle filtering algorithm, building on the CLT of \cite{chopin2004central}.
As discussed above, the following alternative proof of this CLT provides a stronger theoretical foundation for their calculation, because of the problems with the proofs of~\cite{chopin2004central} that we identify in Section~\ref{sec:failureCLT}.

We begin by showing that the random-weight SMC algorithm is weakly consistent at any step $t = 1, \dots, T$, whenever the relevant expectations exist.
\begin{thm}[Convergence in Probability]\label{thm:convinprob}
	For any $\phi_t \in L^1(\spc_t, \pi_t)$, the SMC Algorithm \ref{alg:rwsmc} converges in probability on level $t$ for the weighted particles before resampling
	\begin{equation}
		\sum_{j=1}^{H} \frac{W_t^{j, H}}{\sum_{k=1}^{H} W_t^{k, H}} \phi_t(C_t^{j, H})
		\Pto \E_{\pi_t}[\phi_t] \label{eqn:convpratio}
	\end{equation}
	and for the unweighted particles after resampling
	\begin{equation}
		\frac{1}{H} \sum_{j=1}^{H} \phi_t(\hat C_t^{j, H} ) \Pto \E_{\pi_t}[\phi_t]. \label{eqn:convpresampled}
	\end{equation}
\end{thm}
In addition, we show the following CLT.
\begin{thm}[Central Limit Theorem]\label{thm:clt}
	Let
	\begin{alignat*}{2}
        A_0 &= L^2(\spc_0, \pi_0), \numberthis \label{eqn:assumpat0} \\
		A_l &= \big\{ \phi_l \in L^2(\spc_l, \pi_l) \, &&\big | \,
		c_{l-1} \mapsto \E_{\bar k_l(c_{l-1}, \cdot)}[V_l \phi_l(C_l)] \in A_{l-1},\\
		& &&\phantom{\, \big|} \E_{\tilde \pi_{l}}[|V_l \phi_l(C_l)|^2] < \infty \big\}, \numberthis \label{eqn:assumpat}
	\end{alignat*}
	and assume that the constant function $c_l \mapsto 1 \in A_l$ for all $l \leq t$.
	Then for any $\phi_t \in A_{t}$, the following central limit theorem for the SMC Algorithm \ref{alg:rwsmc} on level $t$ holds
	\begin{align}
		\sqrt{H} \left\{ \sum_{j=1}^{H} \frac{W_t^{j, H}}{\sum_{k=1}^{H} W_t^{k, H}} \phi_t(C_t^{j, H})
		- \E_{\pi_t}[\phi_t] \right \}
		&\Dto \mathcal{N}(0, \Sigma_t(\phi_t)), \label{eqn:clt1} \\
		\frac{1}{\sqrt H} \sum_{j=1}^{H} \left\{ \phi_t(\hat C_t^{j, H}) - \E_{\pi_t}[\phi_t]\right\}
		&\Dto \mathcal{N}(0, \hat{\Sigma}_t(\phi_t))\label{eqn:clt2}
	\end{align}
	with asymptotic variances defined by the following recursion formula:
	let $\hat \Sigma_0(\phi) = \var_{\pi_0}(\phi)$ and for $t > 0$
    \begin{align}
    	\tilde \Sigma_t(\phi) &= \hat \Sigma_{t-1}(\E_{\bar k_t}[\phi]) + \E_{\pi_{t-1}}[\var_{\bar k_t}(\phi)], 
    	\nonumber\\
    	\Sigma_t(\phi) &= \tilde \Sigma_t(V_t (\phi - \E_{\pi_t}[\phi])),\\
    	\hat \Sigma_t(\phi) &= \Sigma_t(\phi) + \var_{\pi_t}(\phi),
    	\nonumber
    \end{align}
    where $\E_{\bar k_t}[\phi] = c_{t-1} \mapsto \E_{\bar k_t(c_{t-1}, \cdot)}[\phi(C_t)]$ and $\var_{\bar k_t}(\phi) = c_{t-1} \mapsto \var_{\bar k_t(c_{t-1}, \cdot)}(\phi(C_t))$.
\end{thm}
Theorem \ref{thm:clt} strongly resembles \cite[Theorem 1]{chopin2004central}, but now generalised to random weights.
Note however that the conditions in \eqref{eqn:assumpat} only require second moments (and not moments of order $2+\delta$ as in \cite{chopin2004central}); this is consistent with the results from \cite{cappe2005inference}.
Similarly, \cite{fearnhead2008particle} considers a random weight particle filter applied to partially observed diffusions and provides a CLT based on the results of \cite{chopin2004central}.

In Theorem \ref{thm:clt}, we describe the asymptotic variance of the random weight SMC method by a recursive formula. 
While this definition clearly illustrates the effect of each step of the SMC method on the variance, the recursive form can be inconvenient to work with.
As shown in \cite{chopin2004central}, we can reformulate this variance into a direct expression.
For this, we define the semigroup $Q$ whose action performs an averaged reweighting at each level of the algorithm which is defined by
\begin{equation}
    Q_t[\phi](C_{t-1}) = \E_{\bar k_t(C_{t-1}, \cdot)}[V_t \phi(C_t)], 
\end{equation}
and $Q_{l+1:t}[\phi] = Q_{l+1} \circ \cdots \circ Q_{t}[\phi]$ for $l+1 \leq t$, and $Q_{t+1:t}[\phi] = \phi$.
Then, we can denote the variance of the SMC algorithm at step $t$ by
\begin{equation} \label{eqn:SMCVarianceSum}
    \Sigma_t(\phi) = \sum_{l=1}^{t} \E_{\tilde \pi_{l}} \left[ \left( V_l Q_{l+1 : t}[\phi - \E_{\pi_t}[\phi]](C_l)  \right)^2 \right].
\end{equation}

\subsection{Effect of the randomness of the weights and the resampling step on the asymptotic variance} \label{subsec:ImpactOfRandomness}
The description of the variance in \eqref{eqn:SMCVarianceSum} also allows us to analyse the effect that the randomness of the weights has on the asymptotic variance of the SMC estimator.
As the averaging performed by the action of $Q$ includes the randomness of the weights, we can use the unbiasedness of the weights \eqref{eqn:WeightConditionalExp} to show that
\begin{align}
    Q_t[\phi](C_{t-1}) & = \E_{\bar k_t(C_{t-1}, \cdot)}\big[ \E[V_t \mid C_t] \phi(C_t) \big]
    \nonumber\\ &
            = \E_{\bar k_t(C_{t-1}, \cdot)}[v_t(C_t) \phi(C_t)].
\end{align}
The semigroup $Q$ is therefore the same as in the non-random-weight case.
The randomness enters the variance only through the weights $V_l, l=1, \dots, t$ that appear in each of the expectations in the direct expression \eqref{eqn:SMCVarianceSum}.
For each summand, we can apply Jensen's inequality for conditional expectations and use the unbiasedness of the weights \eqref{eqn:WeightConditionalExp} to show that
\begin{align}
    \E_{\tilde \pi_{l}} \left[ \left( 
                V_l Q_{l+1 : t}[\phi - \E_{\pi_t}[\phi]](C_l) 
            \right)^2 \right]
        &\geq \E_{\tilde \pi_{l}} \left[ \left( 
                \E[V_l \mid C_l] Q_{l+1 : t}[\phi - \E_{\pi_t}[\phi]](C_l) 
            \right)^2 \right] \nonumber \\
        &=  \E_{\tilde \pi_{l}} \left[ \left( 
                v_l(C_l) Q_{l+1 : t}[\phi - \E_{\pi_t}[\phi]](C_l) 
            \right)^2 \right].
\end{align}
The randomisation of the weights can only increase the asymptotic variance of the SMC estimator.
However, due to the resampling step in Algorithm \ref{alg:rwsmc}, the variance $\Sigma_t$ does not contain expectations of products of random weights.

For comparison, we can consider an estimator which follows Algorithm \ref{alg:rwsmc} but omits resampling, instead multiplying the weights each step.
This is known as sequential importance sampling (SIS) \cite{doucet2009tutorial}.
The resulting estimator has asymptotic variance
\begin{equation} \label{eqn:SISVariance}
    \Sigma_t^{\text{SIS}}(\phi) = 
        \E_{\pi_0} \E_{\bar k_1} \cdots \E_{\bar k_t}\left[ \big(V_1 V_2 \cdots V_t \left( \phi(C_t) - \E_{\pi_t}[\phi]\right) \big)^2 \right],
\end{equation}
see \cite[Section 3]{chopin2004central}.
The variance of the product of random weights in \eqref{eqn:SISVariance} can potentially grow exponentially in the number of steps $t$.
The resampling step of the random-weight SMC algorithm mitigates this accumulation of noise in the importance weights.

The exact impact of the randomisation on the performance of SMC will depend on the specific application.
We now present a simple example of this operation which illustrates the usefulness of SMC in the statistical physics Example~\ref{ex:HS}.

\begin{figure}
    \centering
    \includegraphics[width=\columnwidth]{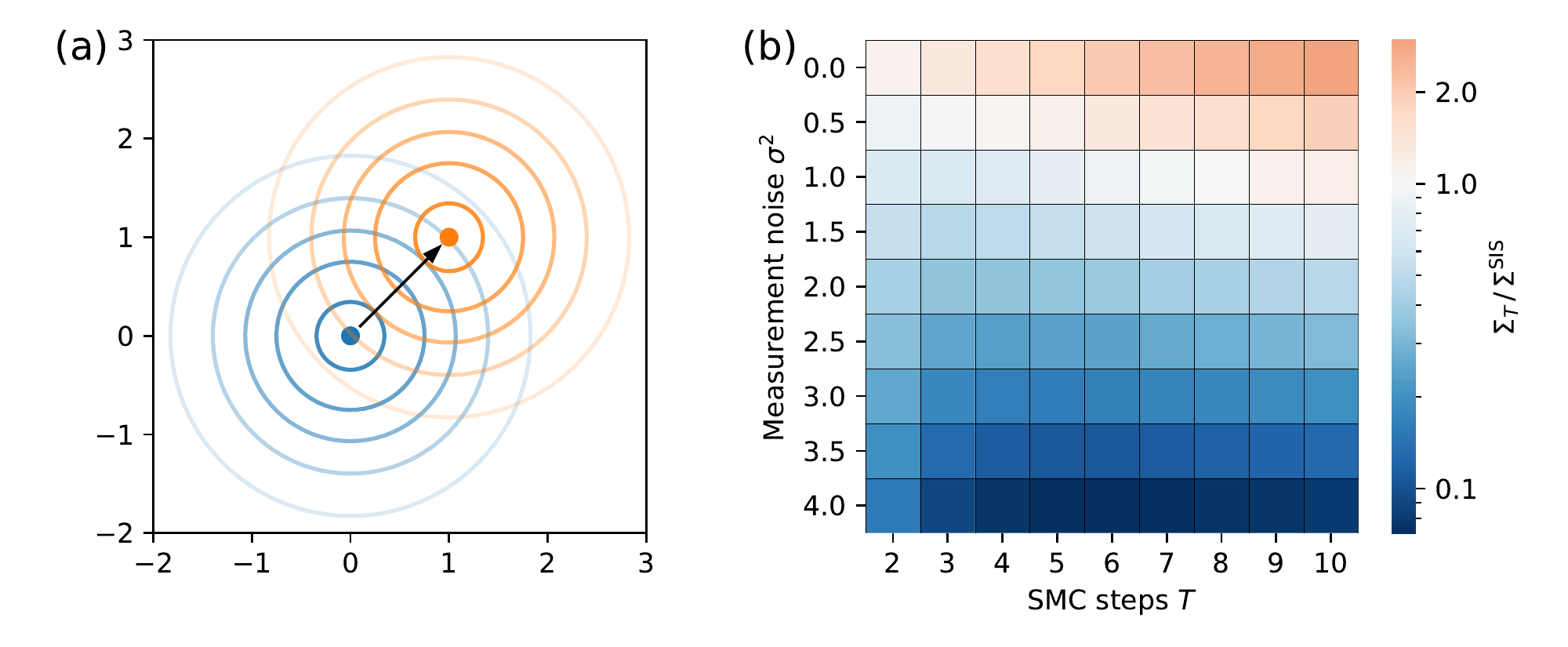}
    \caption{
        Simple random-weight SMC algorithm of Example \ref{ex:RandomWeight}.
        (a) The start and end distribution of the SMC process.
        (b) The quotient $\Sigma_T / \Sigma^{\text{SIS}}$ of the asymptotic variances of an SMC method with $T$ steps and the corresponding sequential importance sampling method without resampling.
        This value is computed for varying number of steps $T$ and variance of the noise of the measured weights $\sigma^2$.
    }
    \label{fig:RandomWeight}
\end{figure}

\begin{examp}  \label{ex:RandomWeight}
    As a very simple example, consider a sequence of $T+1$ unit Gaussian distributions on $\mathbb{R}^2$, that is
    \begin{equation} 
        \pi_t = \mathcal{N}(\mu_t, I) 
    \end{equation}
    with equally-spaced means between (0,0) and (1,1)
    \begin{equation}
        \mu_t = (t/T, t/T),
    \end{equation}
    for $ t=0, \dots, T$, see Figure~\ref{fig:RandomWeight}(a).
    As a quantity of interest, we measure the first component of the $2d$ vector.

    To make the connection to Example~\ref{ex:HS}, we assume that importance weights between steps $w_t = \id \pi_{t} / \id \pi_{t-1}$ are not directly available and need to be estimated.
    In the statistical physics context~\cite{rohrbach2022multilevel}, the analogues of these weights are ratios of partition functions computed by AIS, which leads to unbiased estimates with log-normal distributions~\cite{neal2001annealed}.
    As a simple model for this, we consider a kernel $\bar k_t$ that, for $\hat C_{t-1} \sim \pi_{t-t}$, keeps the position of the particle fixed $C_{t} = \hat C_{t-1}$ and ``measures'' a weight with  log-normal noise of mean $1$
    \begin{equation}
        \hat W_t(C_{t}) \sim \frac{\id \pi_{t}}{\id \pi_{t-1}}(C_{t}) \cdot \log\!\mathcal{N}\left(-\frac{\sigma^2}{2T}, \frac{\sigma}{\sqrt{T}}\right).
    \end{equation}
    The noises are all independent; their distributions are chosen such that the overall importance weight $\hat W = \hat W_1 \cdots \hat W_T$ that appears in the variance of SIS \eqref{eqn:SISVariance} is log-normal with log-variance $\sigma^2$, independent of the number of intermediate steps $T$.
    
    An important feature of this example which is also consistent with the statistical physics setting~\cite{rohrbach2022multilevel} is that the kernels $\bar k_t$ only sample weights and do not modify the positions of the particles.
    This implies that if we set the noise of the weights $\sigma$ to $0$, the first term that appears in the asymptotic variance formula $\Sigma_T$ of the SMC estimator in \eqref{eqn:SMCVarianceSum} 
    \begin{equation}
         \E_{\tilde \pi_{1}} \left[ \left( V_1 Q_{2 : T}[\phi - \E_{\pi_T}[\phi]](C_1)  \right)^2 \right]
            = \E_{\pi_0}\left[v_1 \cdots v_T (\phi - \E_{\pi_T}[\phi)\right]
            = \Sigma_T^{\text{SIS}}(\phi) 
    \end{equation}
    is equal to the variance of the SIS estimator from \eqref{eqn:SISVariance}; the remaining summands are due to the additional error introduced by the resampling steps.
    Thus, we have $\Sigma_T \geq \Sigma_T^{\text{SIS}}$ in the no-noise case, and any improvement of SMC over SIS for $\sigma > 0$ is solely due to the noise-reducing effect of resampling discussed in this subsection.
    
    We now compare the asymptotic variance of the SMC estimator $\Sigma_T$ for varying number of steps $T > 1$ against the variance of the sequential importance sampling estimator $\Sigma^{\text{SIS}}$ (which is equivalent to $\Sigma_1$ here), given different levels of weight noise $\sigma^2$.
    In Figure~\ref{fig:RandomWeight}(b), we plot the ratio of the asymptotic variances $\Sigma_T / \Sigma^{\text{SIS}}$.
    (In this simple example, the expectations in (\ref{eqn:SMCVarianceSum}) can be computed directly.)
    
    As predicted, if the weights are exact ($\sigma^2 = 0$), introducing intermediate steps does not improve the performance.
    Instead, each resampling reduces the effective sample size, leading to an increased variance in the SMC method.
    This is unavoidable (the deleterious effects of resampling could be reduced by adopting a more efficient resampling strategy, for example residual resampling \cite{cappe2005inference}, but this is not our main concern here).
    Instead we focus on systems with significant measurement noise $\sigma^2 > 0$.
    In this case, Fig.~\ref{fig:RandomWeight}(b) shows that the SMC method does outperform the SIS base case, for the reasons discussed above (the asymptotic variance of SMC is lower than $\Sigma^{\text{SIS}}$ because it does not include expectations of products of random weights).
\end{examp}

In practical applications, the distributions might be high dimensional, non-Gaussian, and the noise of the weights not exactly log-Gaussian.
Nevertheless, we have observed a variance reduction similar to Example~\ref{ex:RandomWeight} in an application of RMSMC to the binary hard sphere mixture of Example~\ref{ex:HS} in \cite{rohrbach2022multilevel}.
There, the introduction of one intermediate level while keeping all other parameters fixed has led to a substantially reduced variance of the corresponding estimator.
This was achieved by careful design of the intermediate level, informed by prior understanding of the physics.
In this case, it is not feasible to generalise the approach to $T>2$, contrary to the simple Example~\ref{ex:RandomWeight}, so the large values of $T$ considered in that example are not accessible.  Still, the method improves significantly on the corresponding SIS base case \cite{rohrbach2022multilevel}.

\subsection{Almost sure convergence} \label{subsec:ASConvergence}
We now turn to the almost sure convergence of the algorithm.
Here, the triangular structure of the SMC algorithm causes additional complexity.
Unlike the convergence in distribution and probability that were investigated in the previous two theorems, almost sure convergence depends on the relation of the random variables on different rows of the triangular array, which are independent for our algorithm.
Thus, the canonical strong law of large numbers is not applicable here, see Section \ref{sec:failure}.
To prove convergence, we make use of a strong law of large numbers for triangular arrays of row-wise independent random variables that has been shown by \cite{hu1989strong}.
\begin{thm}[{{\cite[Corollary 1]{hu1989strong}, Strong Law of Large Numbers}}] \label{thm:slln}
	Let 
	\begin{equation}
	    \mathbf{X} = \{(X^{j, H})_{j \leq H}\}_{H \in \N}
	\end{equation}
	be a triangular array of real-valued random variables where for some $\delta > 0$
	\begin{equation}\label{eqn:sllncondition}
		\sup_{j \leq H, H \in \N}\E\left[ |X^{j, H} - \E[X^{j, H}]|^{2 + \delta}\right] < \infty.
	\end{equation}
	If for every $H \in \N$ the random variables in one row $X^{j, H}, \; j = 1, \dots, H,$ are independent, we have
	\begin{equation}
	    \frac{1}{H}\sum_{j=1}^{H} \left( X^{j, H} - \E[X^{j, H}] \right) \longrightarrow 0 \; \text{a.s.}
	\end{equation}
\end{thm}
Similarly to the previous section, we prove almost sure convergence of the random-weight SMC Algorithm \ref{alg:rwsmc} by iteratively applying Theorem \ref{thm:slln}.
\begin{thm}[Almost Sure Convergence]\label{thm:asconvergence}
	Assume that there exists a $\delta > 0$ such that for each $0 < l \leq t$
	\begin{equation}\label{eqn:thmasreq}
		\sup_{\hat{c}_{l-1}} \E_{\bar k_l(\hat{c}_{l-1}, \cdot)}\left[ |V_l|^{2 + \delta} \right] < \infty
	\end{equation}
	for all reachable configurations $\hat{c}_{l-1} \in \spc_{l-1}$ at level $l-1$.
	Let $\phi_t:\mathcal{C}_t \to \R$ be a bounded, measurable function.
	Then on level $t$, the SMC Algorithm \ref{alg:rwsmc} converges almost surely
	\begin{align}
		\sum_{j=1}^{H} \frac{W_t^{j, H}}{\sum_{k=1}^{H} W_t^{k, H}} \phi_t(C_t^{j, H})
			&\longrightarrow \E_{\pi_t}[\phi_t] \;\text{a.s.} \label{eqn:smcasconv1} \\
		\frac{1}{H} \sum_{j=1}^{H} \phi_t(\hat C_t^{j, H}) &\longrightarrow \E_{\pi_t}[\phi_t] \; \text{a.s.}  \label{eqn:smcasconv2} 
	\end{align}
	as $H \to \infty$.
\end{thm}
The prerequisites of Theorem \ref{thm:asconvergence} are in line with previous results on a.s.~convergence of non-random-weight SMC in the literature, e.g.~\cite{crisan2000convergence,crisan2002survey}.
There, a.s.~convergence is established for a particle filter in the context of a state space model where the transition kernel of the signal process is a Feller kernel, and similarly, the result requires bounded weights and quantities of interest.

As a remark on Theorems~\ref{thm:convinprob} and \ref{thm:asconvergence}, we recall from elementary theory that the strong law of large numbers for i.i.d.\ sequences is much harder to prove than the weak law.
We do not expect that the (strong) assumption  \eqref{eqn:thmasreq} is necessary to establish a.s.\ convergence, but it does illustrate that the subtleties required in proofs of a.s.\ convergence cannot be avoided in the SMC context.  
The following section presents an example which shows that a.s.\ convergence does indeed require stronger prerequisites than convergence in probability.
This contrasts with the corresponding results for sums of i.i.d.\ random variables, where the two types of convergence hold in equal generality.
Still, the example of i.i.d.\ sums already illustrates the general expectation that a.s.\ convergence should be harder to prove than CLTs and weak laws of large numbers; our results confirm that this expectation also holds for SMC.
The reason is that a.s.\ convergence in SMC requires an analysis of the convergence of a sequence of SMC computations of increasing particle size $H$, in which the dependencies between different members of this sequence have to be controlled.
In contrast, these dependencies do not matter for convergence in probability.

We end this section with a reminder that the proofs of Theorems~\ref{thm:convinprob}, \ref{thm:clt}, and \ref{thm:asconvergence} are given in Appendix \ref{sec:proofs}.

\section{Limits of almost sure convergence} \label{sec:failure}
Comparing the results in Theorem \ref{thm:convinprob} for convergence in probability with the result for a.s.~convergence in Theorem \ref{thm:asconvergence}, one sees that the latter has much more restrictive prerequisites, in particular requiring bounded quantities of interest $\phi$.
While we do not expect that all assumptions in Theorem~\ref{thm:asconvergence} are necessary, the substantial difference in applicability raises the question whether a.s.~convergence holds less generally for SMC methods.

To address this question, we present a very simple realisation of the non-random-weight SMC Algorithm \ref{alg:smc} where the convergence in probability holds but a.s.~convergence fails.
This result is contrary to the notion that the SMC Algorithm \ref{alg:smc} should converge almost surely whenever its expectation exists, which was suggested by \cite{chopin2004central}.

We first give a short explanation as to why the proof of almost sure convergence is not immediate.
Equations~\eqref{eqn:smcasconv1} and \eqref{eqn:smcasconv2} resemble the classical statement of the strong law of large numbers, which states that for  i.i.d.\ random variables $(X_i)_{i \in \N}$ such that $\E[ X_1 ]$ exists, we have
\begin{equation}\label{eqn:asconvergence1}
    \frac{1}{H} \sum_{j=1}^{H} X_j \longrightarrow \E[X_1] \;\; \text{a.s.}
\end{equation}
as $H \to \infty$.
However, within the context of SMC methods, we do not have a simple sequence of variables, but instead a triangular array.
Unlike other notions of convergence (like convergence in probability and in distribution), the convergences in \eqref{eqn:smcasconv1} and \eqref{eqn:smcasconv2} depend on the relation of random variables that live on different rows.
(Given some large row index $H_0$,  almost-sure convergence requires that we prove statements about the \emph{largest} row average among all rows with $H>H_0$, in order to control the limsup of the sequence.)
For our SMC Algorithms \ref{alg:smc} and \ref{alg:rwsmc}, every row in the triangular array is independent, in particular
\begin{equation}
    \hat C_0^{j, H} \sim \pi_0 \;\; \text{i.i.d.\ for all} \;\; j \leq H, H \in \N.
\end{equation}
As a result, the convergence in \eqref{eqn:smcasconv2} does not hold in such generality as the strong law in \eqref{eqn:asconvergence1}.

We demonstrate this using a counterexample from \cite{zaman1984kolmogorov}. 
This example makes use of heavy-tailed symmetric $\alpha$-stable distributions \cite{samorodnitsky1994stable}.
Fix $1 < \alpha < 2$ and let $S_\alpha$ be a symmetric $\alpha$-stable real-valued random variable with characteristic function
\begin{equation}
    \varphi_{S_\alpha}(t) = e^{-|t|^\alpha}.
\end{equation}
For this random variable, the $p$-th (fractional) moment exists if and only if $p < \alpha$ (assuming $\alpha < 2$).  
This fact is linked to the tail probability which has the following property: for $h>0$ large enough there exists $C_{h,\alpha}>0$ such that for all $x\geq h$ we have
\begin{equation}\label{eqn:propalphastable}
	\mathbb{P}(S_\alpha > x) \geq C_{h,\alpha} x^{-\alpha}.
\end{equation}
\begin{examp}[{{\cite[Example 1]{zaman1984kolmogorov}}}] \label{example}
    Let $\mathbf X$ be a triangular array of independent random variables $X^{j, H} \sim S_\alpha$ i.i.d.~for $1 < \alpha < 2$.
    We can explicitly compute the distribution of sums of $\alpha$-stable random variables as
    \begin{equation}
        \sum_{j=1}^{H} X^{j, H} \sim H^{1/\alpha} S_\alpha.
    \end{equation}
    Hence we can write
        \begin{align}\label{eqn:sumofprobabilities}
        \sum_{H=1}^{\infty} \mathbb P \left( \frac1H \sum_{j=1}^{H} X^{j, H} > c \right) 
        &  \geq  \sum_{H=h}^{\infty} \mathbb P \left( S_\alpha > cH^{1-1/\alpha} \right) 
        \nonumber\\
        & \geq C_{h,\alpha}  \sum_{H=h}^{\infty} c^{-\alpha} H^{1-\alpha} = \infty
    \end{align}
    where we used the bound on the tail probability in \eqref{eqn:propalphastable}, and that the final sum diverges for $\alpha\leq2$.    
    Since all entries in the triangular array $\mathbf X$ are independent, 
    the left hand side of \eqref{eqn:sumofprobabilities} is a sum over the probabilities of independent events and the second Borel-Cantelli lemma yields
    \begin{equation}
        \mathbb P \left( \left\{ \frac{1}{H} \sum_{j=1}^{H} X^{j, H} > c \right\} \text{ i.o.} \right) = 1
    \end{equation}
    where i.o.\ stands for ``infinitely often''.  That is, every random variable $X^{j, H}$ has mean zero, but for any $c > 0$ almost every triangular array has infinitely many rows where the row average is bigger than $c$.  For these arrays, the (random) sequence of row averages does not converge to any limit.
    This means in particular that the row averages do not converge to $\E[S_{\alpha}] = 0$ almost surely.
\end{examp}
The crucial ingredient in this argument is the independence of the rows.
The sum of probabilities in \eqref{eqn:sumofprobabilities} would still be infinite for the sums in \eqref{eqn:asconvergence1}.
However in that case, all of the sums are over the same sequence of random variables, so they are dependent and the second Borel-Cantelli lemma does not apply.  
Indeed, proving the strong law of large numbers for \eqref{eqn:asconvergence1} requires some care with the dependence of the partial sums of the sequence.

\subsection{Consequences for the SMC algorithm} \label{sec:conseq}
Following Example \ref{example}, it is simple to construct an example where the SMC algorithm does not converge almost surely, even on the first step.
Take $\spc_0 = \R$ and let $S_\alpha$ with $1 < \alpha < 2$ be a symmetric $\alpha$-stable random variable and set  $\pi_0 \sim S_\alpha$ and $\phi(x) = x$.
Then, the algorithm is initialised by taking $\hat C_0^{j,H} \sim \pi_0$ i.i.d.~but Example \ref{example} shows that 
\begin{equation} \label{eqn:counterexampleasconv}
	\frac{1}{H} \sum_{j=1}^{H} \phi(\hat C^{j, H}_0)
\end{equation}
does not converge almost surely to $\E_{\pi_0}[\phi] = 0$ as $H\to\infty$.
Hence, Theorem \ref{thm:asconvergence} cannot be generalised to include this case.
However, Theorem \ref{thm:convinprob} shows that \eqref{eqn:counterexampleasconv} does converge in probability to $\E_{\pi_0}[\phi] = 0$.
Indeed, this very simple SMC example shows that the a.s.~convergence requires more demanding assumptions than the convergence in probability.

In this example, the failure of a.s.\ convergence happens already on the first step of the algorithm -- this makes it easy to prove that the convergence fails, because the elements of $\mathbf C_0$ are i.i.d.
More generally, this result illustrates that a.s.\ convergence of triangular arrays cannot be expected if the row elements have heavy-tailed distributions.  Hence, even if explicit counterexamples are not easy to construct, the difficulties of proving a.s.\ convergence are not restricted to the first step of the SMC algorithm.

We note in passing that if the rows of the triangular array were dependent then it might be possible to recover almost-sure convergence.
However, it is not clear to us how to control the dependence of the rows during the resampling step of the SMC algorithm.

\subsection{The role of a.s.~convergence in the CLT of \cite{chopin2004central}}
\label{sec:failureCLT}
This section discusses how the failure of a.s.\ convergence affects the proof of the CLT for (non-random-weight) SMC in \cite[Theorem 1]{chopin2004central}.
The key point is that the CLT at step $t$ is derived in that work under the assumption of a.s.\ convergence on step $t-1$.
In the example of Section \ref{sec:conseq}, a.s.\ convergence fails on step $0$.
We now extend this example to show that it invalidates 
the proof of the CLT in~\cite{chopin2004central}, at step $1$.  (Hence our motivation in this work, to provide a self-contained proof of Theorem~\ref{thm:clt} that does not rely on a.s.\ convergence.)

\paragraph{Example two-step SMC algorithm}
Similar to Section \ref{sec:conseq}, take $\spc_1=\spc_0=\mathbb{R}$ and let $S_\alpha$ with $1 < \alpha < 2$ be a symmetric $\alpha$-stable random variable.
Set $\pi_0 \sim S_\alpha$ and take also $\tilde\pi_1=\pi_0$.
This can be achieved by simply defining the kernel $k_1$ to draw independent samples according to $\pi_0$.
As target distribution for $t=1$ we again choose $\pi_1 = \pi_0$ and therefore $v_1 = \id \pi_1 / \id \tilde \pi_1 \equiv 1$.
Now run Algorithm~\ref{alg:smc}.
After the first correction step, the particles $C^{j,H}_1\sim S_\alpha$ are i.i.d.\ and $W^{j,H}_1=1$.
The next step is to resample the particles using multinomial resampling
\begin{equation}
	(C_1^{j, H}, W_1^{j, H})_{j \leq H} \longrightarrow (\hat C_1^{j, H})_{j \leq H}.
\end{equation}
Consider the quantity of interest 
\begin{equation}\label{eqn:def-phi}
    \phi(c) = \mathrm{sign}(c) |c|^{1/2}.
\end{equation}
and its estimator
\begin{equation}
    \hat{F}_1^H = \frac{1}{H} \sum_{j=1}^{H} \phi( \hat C_1^{j, H} ) .
\end{equation}
This SMC algorithm satisfies the requirements of Theorem~\ref{thm:clt} so the CLT holds for $\hat{F}_1^H$ as $H\to\infty$.

Let us now examine this example in the context of \cite[Theorem 1]{chopin2004central}.
Note that in the slightly different formulation of that work, the initialisation and first mutation step coincide: this makes no difference for the current example because the kernel $k_1$ draws independent samples and so can be equivalently interpreted as the initialisation step of the algorithm of~\cite{chopin2004central}, after which the algorithms are equivalent (apart from the notational difference that the initialisation happens at $t=1$ here).
To establish that the example algorithm is within the scope of \cite[Theorem 1]{chopin2004central}, we now need to check that $c_1 \mapsto 1, \phi \in \Phi_1$ where
\begin{equation}
    \Phi_1 := \left\{ \psi \in L^2(\mathcal{C}_1, \tilde \pi_1) \, \middle| \, \E_{\tilde \pi_1}[|v_1(C_1) \psi(C_1)|^{2+\delta}] < \infty \right\}
\end{equation}
for some $\delta>0$.
Since $v_1 \equiv 1$, it is trivial that the unit function $c_1 \mapsto 1 \in \Phi_1$.
Furthermore, we have
\begin{equation}
    \E_{\tilde \pi_1}[|v_1(C_1) \phi(C_1)|^{2+\delta}]
        = \E_{\tilde \pi_1}[|\phi(C_1)|^{2+\delta}]
        = \E_{\tilde \pi_1}[|C_1|^{1 + \delta/2}] < \infty
    \label{equ:E-phi-delta}
\end{equation}
for $\delta < 2(\alpha - 1)$, where we have used the definition of $\phi$ from \eqref{eqn:def-phi} for the second equality.
Hence the scheme satisfies the stated requirements of \cite[Theorem 1]{chopin2004central}.

The strategy of \cite{chopin2004central} is to establish a CLT for $\hat{F}_1^H$ by applying a CLT for triangular arrays with independent elements.
The elements of $\hat{\mathbf{C}}_1$ are not independent but they are conditionally independent given the particles $\mathbf C_1$ before resampling.
To exploit this, one establishes that for almost every realisation of $\mathbf C_1$, a conditional CLT holds for $\hat{\mathbf{C}}_1$.
Under certain technical conditions, this can be used to prove a CLT for the full distribution of $\hat{\mathbf{C}}_1$.
This strategy is appealing and intuitive, but the technical conditions require some care.

The remainder of this section outlines these conditions, and shows that they are violated for the example two-step SMC algorithm described above.
Specifically, the CLT for $\hat{F}_1^H$ requires \cite[Lemma A.3]{chopin2004central}, whose proof fails in our example.
A similar example can be constructed for which the (analogous) proof of \cite[Lemma A.1]{chopin2004central} breaks down. 
However, we emphasise once more that this is a statement about the proof and not about the validity of the CLT in \cite{chopin2004central}: the two-step example obeys Theorem~\ref{thm:clt} which is equivalent to the CLT of \cite{chopin2004central} in this case.

\subsection{Technical conditions for the CLT: two-step example}
We consider aspects of the proofs of Lemmas A.1 and A.3 of \cite{chopin2004central}, which are analogous to each other.
To simplify notation, we denote the $H$-th row of the triangular array by $\mathbf C_1$ by $\mathbf C_1^H$.
The conditional expectation $\E[\phi(\hat C_1^{j, H}) \, | \, \mathbf C_1]$ depends on all variables in the row $\mathbf C_1^H$.
Define
\begin{equation}
	\sigma^2(\mathbf C_1^H) 
		= \var \left( \phi(\hat C_1^{j, H}) \, \middle| \, \mathbf C_1^H\right)
		= \frac{1}{H} \sum_{k=1}^{H} \phi(C_1^{k, H})^2 - \left(\frac{1}{H} \sum_{k=1}^{H} \phi(C_1^{k, H}) \right)^2
    \label{equ:def-sigma2}
\end{equation}
where we used that $v_1 \equiv 1$ in the second equality.
Define also
\begin{equation}
	\nu_\delta(\mathbf C_1^H)
	= \E \left[ \left| \phi(\hat C_1^{j, H}) - \E[ \phi(\hat C_1^{j, H}) \, | \, \mathbf C_1^H] \right|^{2+\delta} \, \middle| \, \mathbf C_1^H\right].
    \label{equ:def-nu}
\end{equation}
Note that $\sigma^2(\mathbf C_1^H)$ and $\nu_\delta(\mathbf C_1^H)$ do not depend on $j=1, \dots, H$.

Among the necessary technical conditions required for the CLT are two a.s.\ convergences as $H\to\infty$: for some $\delta>0$ then
\begin{equation}
	\nu_\delta(\mathbf C_1^H) \to \nu_\delta
	\;\; \text{and} \;\;
	\sigma^2(\mathbf C_1^H) \to \sigma^2 \qquad \mathrm{a.s.}
\end{equation}
with finite limits $\nu_\delta, \sigma^2 < \infty$.
Their convergence ensures that a Lyapunov condition holds almost surely, which in turn proves a conditional CLT for almost all realisations of $\mathbf C_1$.
For the two-step example algorithm, we show that these quantities do not have a.s.\ convergence to any limit.

By construction, we have $C_1^{j, H} \sim S_\alpha$ i.i.d.
We have from \eqref{equ:E-phi-delta} that $\E_{\tilde \pi_1}[|\phi|^{(2+\delta)}] < \infty$.
Then Theorem \ref{thm:slln} shows that
\begin{equation}
	\frac{1}{H} \sum_{j=1}^{H} \phi(C_1^{j, H}) \rightarrow \E_{\tilde \pi_1}[\phi] = 0 \;\; \text{a.s.}
\end{equation}
which, using independence of the elements of $\mathbf C_1$ and the Borel-Cantelli lemma, yields for any $c > 0$
\begin{equation}
	\mathbb P \left( \left\{ \frac{1}{H} \sum_{j=1}^{H} \phi(C_1^{j, H}) > c \right\} \text{ i.o.} \right)
	 = 0.
	  \label{equ:Pio1}
\end{equation}
Also, for any $c > 0$
\begin{equation}
    \mathbb P \left( \left\{ \frac{1}{H} \sum_{j=1}^{H} \phi(C_1^{j, H})^2 > c \right\} \text{ i.o.} \right) 
        \geq  \mathbb P \left( \left\{ \frac{1}{H} \sum_{j=1}^{H} C_1^{j, H} >  c \right\} \text{ i.o.} \right)
        = 1,
    \label{equ:Pio2}
\end{equation}
where the first equality uses the definition of $\phi$ and the second is from Example \ref{example}.
Combining \eqref{equ:def-sigma2},\eqref{equ:Pio1},\eqref{equ:Pio2}, we get $\mathbb P(\{ \sigma^2(\mathbf C_1^H) > \text{i.o.}\}) = 1$ which proves that $\sigma^2(\mathbf C_1^H)$ does not converge almost surely to any limit as $H \to \infty$.
The same argument applies to $\nu_\delta(\mathbf C_1^H)$:
using $v_1 \equiv 1$ and H\"older's inequality in \eqref{equ:def-nu} yields $\nu_\delta(\mathbf C_1^H) \geq (\sigma^2(\mathbf C_1^H))^{2/(2+\delta)}$, so
\begin{equation}
	\mathbb P(\{ \nu_\delta(\mathbf C_1^H) > c\} \; \text{i.o.})
		\geq P(\{ \sigma^2(\mathbf C_1^H) > c^{(2+\delta)/2}\} \; \text{i.o.} )
		= 1.
\end{equation}
One sees that neither $\sigma^2(\mathbf C_1^H)$ and $\nu(\mathbf C_1^H)$ satisfy a.s.\ convergence $H \to \infty$ (as asserted above).
This means that we cannot establish a Lyapunov condition that ensures the conditional convergence of $\hat{F}_1^H$, and the CLT must be established in a different way.

\section*{Acknowledgements}
We thank Nicolas Chopin for helpful advice on methods for proving SMC convergence.
We thank Robert Scheichl, Jan Johannes, Sumeetpal Singh, and Sam Power for helpful discussions and comments on the manuscript, and Nigel Wilding for discussions about statistical physics applications of Monte Carlo methods.
This project was supported by the Leverhulme Trust through project grant RPG-2017-203.

\appendix

\section{Proofs} \label{sec:proofs}
This appendix contains the proof for Theorem \ref{thm:asconvergence} and discusses how Theorem \ref{thm:convinprob} and \ref{thm:clt} are proven following \cite{cappe2005inference}.

\subsection{Almost sure convergence} \label{sec:proofs_as}
We start with the proof of the almost sure convergence result for the SMC algorithm in Theorem \ref{thm:asconvergence}.
The central idea of the proof is the fact that the particles at each step are conditionally independent given the previous step.
This allows us to apply the strong law of large numbers given in Theorem \ref{thm:slln} iteratively, proving convergence of the current step given convergence of the previous one.
This approach is similar to the one taken in \cite{chopin2004central} to prove a CLT.

\begin{proof}[Proof of Theorem \ref{thm:asconvergence}]
    For $t=0$, the convergence of
    \begin{equation}
        \frac{1}{H} \sum_{j=1}^{H} \phi_0(\hat C_0^{j, H})
            \longrightarrow \E_{\pi_0}[\phi] \; \text{a.s.}
    \end{equation}
    for bounded $\phi_0$ follows directly from Theorem \ref{thm:slln}, since all moments of $\phi_0(\hat C_0^{j, H})$ are in particular uniformly bounded.

    For $t>0$, we prove the a.s.\ convergence of SMC in two parts.
    First, we prove the convergence of the ratio estimator \eqref{eqn:smcasconv1} of the weighted particles after the mutation and correction step in Algorithm \ref{alg:rwsmc}, using convergence on the previous level $t-1$.
    Subsequently, we show the convergence of the estimator of the resampled particles \eqref{eqn:smcasconv2}.

    We start with the convergence of \eqref{eqn:smcasconv1}.
    Assume that we have shown the results of Theorem \ref{thm:asconvergence} for all bounded functions $\phi_{t-1}$ on level $t-1$.
    Using this, we show that the non-ratio estimator
    \begin{equation} \label{eqn:nonratioest}
        F_t^H := \frac{1}{H} \sum_{j=1}^{H} V_t^{j,H} \phi_t(C_t^{j,H})
    \end{equation}
    converges a.s., that is $\mathbb P \left( F_t^H \to \E_{\pi_t}[\phi_t] \right) = 1$.
    By making use of the conditional independence of the elements of $\mathbf C_t$, we can apply Theorem \ref{thm:slln} conditional on the previous step $\hat{\mathbf C}_{t-1}$.
    This allows us to rewrite the conditional expectation as a function from $\mathcal{C}_{t-1} \to \R$ for which the induction hypothesis holds.
    
    Using the Markov transition kernel $\bar k_t$, we can compute the conditional expectation of the row average $F_t^H$ as
    \begin{equation} \label{eqn:asconvdefcondexp}
        \E[F_t^H \,|\, \hat{\mathbf C}_{t-1}]
            =  \frac{1}{H} \sum_{j=1}^{H} \E[ V_t^{j, H} \phi_t(C^{j, H}_t) \,|\, \hat{\mathbf C}_{t-1}]
            =  \frac{1}{H} \sum_{j=1}^{H} \E_{\bar k_t(\hat{C}_{t-1}^{j, H}, \cdot)}[ V_t \phi_t(C_t) ],
    \end{equation}
    where the second equality uses that $C_t^{j,H}$ is generated in the algorithm by mutation of $\hat{C}_{t-1}^{j,H}$ with the kernel $\bar k_t$.
    Now let $(\mathbf C_t, \mathbf V_t) | \hat{\mathbf c}_{t-1} \sim \bar k_t(\hat{\mathbf c}_{t-1}, \cdot)$ be the result of applying the Markov transition kernel $\bar k_t$ to a realisation $\hat{\mathbf c}_{t-1}$ of $\hat{\mathbf C}_{t-1}$.
    By construction, $(\mathbf C_t, \mathbf V_t) | \hat{\mathbf c}_{t-1}$ is an array of independent random variables and by assumption \eqref{eqn:thmasreq} we have
    \begin{equation} \label{eqn:upperboundasconv}
        \sup_{j \leq H, H \in \N} \E_{\bar k_t(\hat c_{t-1}^{j, H}, \cdot)}\left[ | V_t \phi_t(C_t) |^{2+\delta}\right]
            \leq \|\phi_t^{2 + \delta}\|_{\infty} \sup_{\hat c_{t-1}} \E_{\bar k_t(\hat c_{t-1}, \cdot)}\left[ | V_t |^{2+\delta} \right] 
            < \infty.
    \end{equation}
    Therefore, we can apply Theorem \ref{thm:slln} to prove the convergence of $F_t^H|\hat{\mathbf c}_{t-1}$ to its conditional expectation \eqref{eqn:asconvdefcondexp} given $\hat{\mathbf C}_{t-1} = \hat{\mathbf c}_{t-1}$, that is
    \begin{equation} \label{eqn:asconvfixedrealisation}
        \mathbb P\left( F_t^H | \hat{\mathbf c}_{t-1} - \E[F_t^H \,|\, \hat{\mathbf C}_{t-1}=\hat{\mathbf c}_{t-1}] \to 0 \right) = 1.
    \end{equation}
    Since the convergence in \eqref{eqn:asconvfixedrealisation} holds for any realisation $\hat{\mathbf c}_{t-1}$, we have computed a version of the conditional expectation
    \begin{equation} \label{eqn:asconvfth}
        \E[ \mathds{1}_{\{ F_t^H - \E[F_t^H \,|\, \hat{\mathbf C}_{t-1}]  \to 0 \} } \,|\, \hat{\mathbf C}_{t-1}]
        = \E_{\bar k_t(\hat{\mathbf C}_{t-1}, \cdot)}[\mathds{1}_{\{ F_t^H - \E[F_t^H \,|\, \hat{\mathbf C}_{t-1}]  \to 0\} } ]
        = 1,
    \end{equation}
    where $\mathds{1}_{A}$ is the indicator function of the set $A$.

    As we have shown in \eqref{eqn:upperboundasconv}, the function $\phi_{t-1}(c_{t-1}) := \E_{\bar k_t(c_{t-1}, \cdot)}[V_t \phi_t(C_t)]$ is bounded and we can apply the convergence \eqref{eqn:smcasconv2} on level $t-1$ to the conditional expectation in \eqref{eqn:asconvdefcondexp} yielding
    \begin{equation} \label{eqn:asconvcondexp}
        \mathbb P\left(\E[F_t^H \,|\, \hat{\mathbf C}_{t-1}] \to \E_{\pi_{t-1}}[\phi_{t-1}(C_{t-1})]\right) 
            = 1.
    \end{equation}
    By definition of the SMC algorithm we have also
    \begin{equation}
        \E_{\pi_{t-1}}[\phi_{t-1}(C_{t-1})] = \E_{\pi_{t-1}}[ \E_{\bar k_t(C_{t-1}, \cdot)}[V_t \phi_t(C_t)]]
            = \E_{\pi_t}[\phi_t].
    \end{equation}
    Hence the convergence point of the limit in \eqref{eqn:asconvcondexp} coincides with 
    the level-$t$ average of the SMC algorithm.
    
    To prove a.s.\ convergence of \eqref{eqn:nonratioest}, we decompose the corresponding expectation into two steps
    \begin{align}
        \mathbb P( F_t^H \to \E_{\pi_t}[\phi_t])
            &= \E[ \mathds{1}_{\{ F_t^H \to \E_{\pi_t}[\phi_t] \} } ] \nonumber \\
            &\geq \E\left[ \mathds{1}_{\{ F_t^H - \E[F_t^H \,|\, \hat{\mathbf C}_{t-1}] \to 0 \} }
                 \cdot \mathds{1}_{ \{\E[F_t^H \,|\, \hat{\mathbf C}_{t-1}] - \E_{\pi_t}[\phi_t] \to 0 \} } \right] \nonumber \\
            &= \E\left[ \E[  \mathds{1}_{\{ F_t^H - \E[F_t^H |\hat{\mathbf C}_{t-1}] \to 0 \} } \,|\, \hat{\mathbf C}_{t-1}]  
                \cdot \mathds{1}_{ \{\E[F_t^H \,|\, \hat{\mathbf C}_{t-1}] - \E_{\pi_t}[\phi_t] \to 0 \} } \right].
    \end{align}
    Using \eqref{eqn:asconvfth} and \eqref{eqn:asconvcondexp}, we have
    \begin{equation}
         \E\left[ \E[  \mathds{1}_{\{ F_t^H - \E[F_t^H |\hat{\mathbf C}_{t-1}] \to 0 \} } \,|\, \hat{\mathbf C}_{t-1}]  
                \cdot \mathds{1}_{ \{\E[F_t^H \,|\, \hat{\mathbf C}_{t-1}] - \E_{\pi_t}[\phi_t] \to 0 \} } \right]
            = 1
    \end{equation}
    and therefore $\mathbb P(F_t^H \to \E_{\pi_t}[\phi_t]) = 1$ which proves the almost sure convergence of the weighted averages on level $t$.
    
    Since the constant function $c_t \mapsto 1$ is bounded, the previous argument also shows that the weight averages converge as
    \begin{equation}
        \mathbb P\left( \frac{1}{H} \sum_{j=1}^{H} V_t^{j,H} - \E_{\pi_t}[1] \to 0 \right) = 1.
    \end{equation}
    Therefore, we also have shown that
    \begin{equation}
    	\mathbb P\left( \sum_{j=1}^{H} \frac{V_t^{j, H}}{ \sum_{k=1}^{H} V_t^{k, H}} \phi_t(C_t^{j, H}) \to \E_{\pi_t}[\phi_t] \right) = 1
    	\label{eqn:asconvofratio}
    \end{equation}
    which proves the convergence of the ratio estimator \eqref{eqn:smcasconv1} (which does not depend on the normalisation of the weights $W_t \propto V_t$).
    This completes the first part of the proof.
    
    We now prove the second part of the theorem, that is the a.s.\ convergence in \eqref{eqn:smcasconv2} for the resampled triangular array $\hat{\mathbf C}_{t}$.
    The proof follows the same steps as above.
    Let $m$ be the rowwise multinomial resampling kernel.
    Let $\hat{\mathbf C}_t | ({\mathbf c}_t,{\mathbf v}_t) \sim m(\mathbf c_t,{\mathbf v}_t; \cdot)$ be the result of applying the multinomial resampling kernel to the rows of a realisation $(\mathbf c_t, \mathbf v_t)$ of the array of particles and weights $(\mathbf C_t, \mathbf V_t)$.
    We need to prove the a.s.\ convergence of the resampled array
    \begin{equation}
        \hat F_t^H := \frac{1}{H} \sum_{j=1}^{H} \phi_t(\hat C_t^{j,H}).
    \end{equation}
    Since $\phi_t$ is bounded and $\hat{\mathbf C}_t | ({\mathbf c}_t,{\mathbf v}_t)$ is an array of independent random variables, we can apply the strong law of large numbers from Theorem \ref{thm:slln} to get
    \begin{equation}
        \mathbb P \left(\hat{F}^H_t|({\mathbf c}_t,{\mathbf v}_t)
                - \E[\hat{F}^H_t \,|\, ({\mathbf C}_t,{\mathbf V}_t)=({\mathbf c}_t,{\mathbf v}_t)] \to 0 \right) = 1
    \end{equation}
    for any realisation $({\mathbf c}_t,{\mathbf v}_t)$ which, as before, gives us
    \begin{equation}
        \E[  \mathds{1}_{\{ \hat F_t^H - \E[\hat F_t^H \,|\, \mathbf C_t, \mathbf V_t]  \to 0 \} } \,|\,  \mathbf C_t, \mathbf V_t]
            = \E_{m( \mathbf C_{t}, \mathbf V_{t}; \cdot)}[  \mathds{1}_{\{ \hat F_t^H - \E[\hat F_t^H \,|\,  \mathbf C_t, \mathbf V_t]  \to 0 \} } ]
            = 1.
    \end{equation}
    The definition of multinomial resampling is that each element of the resampled array is drawn independently from the corresponding row of the original array, with probabilities given by the normalised row weights.
    Hence (a version of) the conditional expectation of $\hat F_t^H$ is
    \begin{equation}
        \E[ \hat F_t^H \,|\,  \mathbf C_{t}, \mathbf V_{t}] 
            = \sum_{j=1}^{H} \frac{V_t^{j, H}}{\sum_{k=1}^{H} V_t^{k, H}} \phi_t(C_t^{j, H})
    \end{equation}
    which is the ratio estimator from \eqref{eqn:asconvofratio} which converges a.s.\ to $\E_{\pi_t}[\phi_t]$.
	Therefore,
    \begin{align}
        \mathbb P(\hat F_t^H \to \E_{\pi_t}[\phi_t]) 
            &= \E[ \mathds{1}_{\{ \hat F_t^H \to  \E_{\pi_t}[\phi_t] \}}] \nonumber \\
            &\geq \E\left[ \mathds{1}_{\{ \hat F_t^H - \E[\hat F_t^H \,|\, \mathbf C_t, \mathbf V_t]  \to 0 \} }  
                \cdot \mathds{1}_{\{ \E[\hat F_t^H \,|\, \mathbf C_t, \mathbf V_t] - \E_{\pi_t}[\phi_t] \to 0 \} }
                \right] \nonumber \\
            &= 1.
    \end{align}
    This proves the a.s.\ convergence in \eqref{eqn:smcasconv2}.
\end{proof}

\subsection{Convergence in probability and central limit theorem}
The proofs of Theorem \ref{thm:convinprob} and Theorem \ref{thm:clt} are based on \cite{cappe2005inference}, which provides a central limit theorem for an SMC method applied to hidden Markov Models.
We apply their results to prove convergence in probability and a central limit theorem for our formulation of the random-weight SMC Algorithm \ref{alg:rwsmc}.
In particular, we show that the randomisation of the weights does not add any additional complications.
Fundamentally, the approach taken by \cite{cappe2005inference} is also similar to the previous section and to \cite{chopin2004central}, as it develops the result by iteratively following the steps of the algorithm and making use of conditional independence.
However unlike \cite{chopin2004central}, it does not rely on almost sure convergence as the preliminary result for the CLT, which appears to severely limit the applicability of those proofs.

First, we consider the convergence results for triangular arrays of conditionally independent variables that underlie the corresponding convergence results for SMC.
They are the equivalent of the strong law of large numbers for triangular arrays used for the proof in Appendix \ref{sec:proofs_as}.
We start with a preliminary definition.
\begin{defin}[Bounded in Probability]
    We call a sequence of random variables $\{Z^H\}_{H \in \N}$ bounded in probability if
    \begin{equation}
        \lim_{C \to \infty} \sup_{H \in \N} \mathbb P(|Z^H| \geq C) = 0.
    \end{equation}
\end{defin}
Any sequence of random variables that converges in distribution to $Z^H \Dto Z \in L^1$ is in particular bounded in probability.
We call a triangular array of random variables $\mathbf X$ conditionally independent with respect to a sequence of sub-$\sigma$-algebras $\{\mathcal F^{H}\}_H$ if the elements of the $H$-th row of the array $\{X^{j, H}, j=1, \dots, H\}$ are conditionally independent given $\mathcal{F}^H$.
For conditionally independent triangular arrays of random variables, one has a weak law of large numbers, given by the following theorem.
\begin{thm}[{{\cite[Prop.\ 9.5.7]{cappe2005inference}}}] \label{thm:condwlln}
    Let $\mathbf X$ be a triangular array of random variables and $\{\mathcal F^H \}_H$ a sequence of sub-$\sigma$-algebras of $\mathcal F$.
    Assume that the following conditions hold.
    \begin{enumerate}[label=(\roman*)]
        \item The triangular array $\mathbf X$ is conditionally independent given $\{\mathcal F^H \}_H$; also, for any $H$ one has for $j=1,2,\dots,H$ that $\E\left[|X^{j, H}| \, \middle | \, \mathcal{F}^H\right] < \infty$ with probability $1$.
        \item The sequence $\left \{ \sum_{j=1}^{H} \E\left[ |X^{j, H}| \, \middle | \, \mathcal F^H\right] \right\}_{H \in \N}$ is bounded in probability.
        \item For all $\eps > 0$,
        \begin{equation} \label{eqn:condwlln_eps}
            \sum_{j=1}^{H} \E\left[ |X^{j, H}| \mathds{1}_{\{ |X^{j, H}| \geq \eps \}} \, \middle | \, \mathcal F^H\right]
                \Pto 0.
        \end{equation}
    \end{enumerate}
    Then
    \begin{equation}
        \sum_{j=1}^{H} \left\{ X^{j, H} - \E\left[ X^{j, H} \, \middle | \, \mathcal F^H\right] \right\}
            \Pto 0.
    \end{equation}
\end{thm}
There is a corresponding central limit theorem for conditionally independent triangular arrays, as follows.
\begin{thm}[{{\cite[Prop.\ 9.5.12]{cappe2005inference}}}]\label{thm:condclt}
    Let $\mathbf X$ be a triangular array of random variables and $\{\mathcal F^H \}_H$ a sequence of sub-$\sigma$-algebras of $\mathcal F$.
    Assume that the following conditions hold.
    \begin{enumerate}[label=(\roman*)]
        \item The triangular array $\mathbf X$ is conditionally independent given $\{\mathcal F^H \}_H$; also, for any $H$ one has for $j=1,2,\dots,H$ that $\E\left[(X^{j, H})^2 \, \middle | \, \mathcal{F}^H\right] < \infty$ with probability $1$.
        \item There exists a constant $\sigma^2$ such that
        \begin{equation}
            \sum_{j=1}^{H} \left\{ \E\left[ (X^{j, H})^2 \, \middle | \, \mathcal{F}^H \right]
                - \left( \E \left[ X^{j, H} \, \middle | \, \mathcal{F}^H \right] \right)^2 \right\}
                \Pto \sigma^2.
        \end{equation}
        \item For all $\eps > 0$
        \begin{equation}
            \sum_{j=1}^{H} \E \left[ (X^{j, H})^2 \mathds{1}_{\{|X^{j, H}| \geq \eps \}} \, \middle | \, \mathcal{F}^H \right]
                \Pto 0.
        \end{equation}
    \end{enumerate}
    Then for any $u \in \R$
    \begin{equation}
        \E \left[ \exp \left( i u \sum_{j=1}^H \left\{ X^{j, H} - \E[X^{j, H} \, | \, \mathcal F^{H}] \right\} \right) \, \middle | \, \mathcal{F}^H \right]
            \Pto \exp(-\sigma^2 u^2/2).
    \end{equation}
\end{thm}
To prove the results of Theorem \ref{thm:convinprob} and \ref{thm:clt}, we follow the same steps as the proofs of the corresponding results for the SMC algorithm for hidden Markov models presented in \cite[Chapter 9]{cappe2005inference}.
At each level of the SMC algorithm, the triangular array of random variables is conditionally independent given the previous step.
To establish convergence in probability and a central limit theorem, we iteratively apply Theorem \ref{thm:condwlln} or Theorem \ref{thm:condclt} conditional on the previous step, respectively.
The corresponding results for an SMC algorithm for hidden Markov models is given in \cite[Theorem 9.4.5]{cappe2005inference}.
For completeness, we present a self-contained version of the proofs for our version of the algorithm using the notation of the main text below.

\subsubsection{Proof of Theorem \ref{thm:convinprob}}
We start with the proof of Theorem \ref{thm:convinprob} on the convergence of probability of the SMC algorithm.
It follows the same steps as the proofs of Theorem 9.3.5 and Theorem 9.2.9 of \cite{cappe2005inference}.
\begin{proof}[Proof of Theorem \ref{thm:convinprob}]
    At $t=0$, we start with an array of independent random variables $\hat C_0^{j, H} \sim \pi_0$, and therefore by the weak law of large numbers we have
    \begin{equation}
        \frac{1}{H} \sum_{j=1}^{H} \phi_0(\hat C_0^{j, H}) \Pto \E_{\pi_0}[\phi_0]
    \end{equation}
    for all $\phi_0 \in L^1(\mathcal{C}_0, \pi_0)$.

    For $t > 0$, we prove the convergence by induction.
    Similar to the proof of a.s.\ convergence in Appendix \ref{sec:proofs_as}, we split the proof in two steps, starting with the convergence of the ratio estimator (\ref{eqn:convpratio}) and then looking at the convergence of the estimator of the resampled array (\ref{eqn:convpresampled}).
    
    To prove the convergence of (\ref{eqn:convpratio}), assume that on level $t-1$ for the triangular array $\hat{\mathbf{C}}_{t-1}$, we have
    \begin{equation} \label{eqn:convpinduction}
        \frac{1}{H} \sum_{j=1}^{H} \phi_{t-1}(\hat C_{t-1}^{j, H}) \Pto \E_{\pi_{t-1}}[\phi_{t-1}] 
    \end{equation}
    for all $\phi_{t-1} \in L^1(\mathcal{C}_{t-1}, \pi_{t-1})$.
    Let
    \begin{equation}
        \mathcal{F}_{t-1}^H = \sigma( \hat C_{t-1}^{j, H}, j=1, \dots, H )
    \end{equation}
    be the $\sigma$-algebra of the $H$-th row of the previous step, before the mutation of the particles.
    By construction, $(\mathbf{C}_t, \mathbf{V}_t)$ is a triangular array of conditionally independent random variables given $\{\mathcal{F}_{t-1}^H\}_H$.
    Let 
    \begin{equation}
        \mu(c_{t-1}) := \E_{\bar k_t(c_{t-1}, \cdot)}[V_t \phi_t(C_t)].
    \end{equation}
    Observe that $|\mu(c_{t-1})| \leq E_{\bar k_{t}(c_{t-1}, \cdot)}[V_t |\phi_t(C_t)|]$ since $V_t > 0$.
    Using this and the unbiasedness of the weights $V_t$, one obtains
    \begin{equation}\label{eqn:boundonmu}
        \E_{\pi_{t-1}}[|\mu(C_{t-1})|] \leq \E_{\tilde{\pi}_t}[V_t |\phi_t(C_t)|] = \E_{\pi_t}[|\phi_t|] < \infty
    \end{equation}
    where the last inequality follows because $\phi_t \in L^1(\mathcal{C}_t, \pi_t)$ by assumption.
    By construction, $\mu$ is the conditional expectation of the weighted particle expectation
    \begin{equation}
        \E[V_t^{j, H} \phi_t(C_t^{j, H}) \,|\, \mathcal{F}_{t-1}^H]
            = \mu(\hat C_{t-1}^{j, H}).
    \end{equation}
    Following from \eqref{eqn:boundonmu}, $\mu \in L^1(\mathcal{C}_{t-1}, \pi_{t-1})$ and we can apply the convergence \eqref{eqn:convpinduction} on level $t-1$
    \begin{equation} \label{eqn:convofcondexp}
        \frac{1}{H} \sum_{j=1}^{H} \E[V_t^{j, H} \phi_t(C_t^{j, H}) \,|\, \mathcal{F}_{t-1}^H]
            = \frac{1}{H} \sum_{j=1}^{H} \mu(\hat C_{t-1}^{j, H})
            \Pto \E_{\tilde \pi_{t}}[V_t \phi_t(C_t)] 
            = \E_{\pi_t}[\phi_t].
    \end{equation}
    This establishes that mutation and reweighting provides an estimate of $\phi_t$ with the correct conditional expectation.

    As the next step, we show that the SMC estimator itself converges in probability to its conditional expectation:
    \begin{equation} \label{eqn:convpproofconv1}
        \frac{1}{H} \sum_{j=1}^{H} V_t^{j, H} \phi_t(C_t^{j, H}) 
                - \frac{1}{H} \sum_{j=1}^{H} \E[V_t^{j, H} \phi_t(C_t^{j, H}) \,|\, \mathcal{F}_{t-1}^H]
            \Pto 0.
    \end{equation}
    For this, we apply Theorem \ref{thm:condwlln} to the conditionally independent triangular array $\mathbf X_t$ with elements
    \begin{equation}
        X_t^{j, H} := \frac{1}{H} V_t^{j, H} \phi_t(C_t^{j, H}).
    \end{equation}
    Condition (i) of Theorem \ref{thm:condwlln} holds since the conditional expectation is almost surely finite following \eqref{eqn:boundonmu}.
    It remains to check conditions (ii) and (iii).
    Following the same argument as for (\ref{eqn:convofcondexp}), we have
    \begin{equation}
        \sum_{j=1}^{H} \E\big[|X_t^{j, H}| \, \big| \, \mathcal{F}_{t-1}^H\big] \Pto \E_{\pi_{t}}[|\phi_t|]
    \end{equation}
    which implies boundedness in probability, so condition (ii) holds.
    For any $\eps > 0$, we have
    \begin{equation} \label{eqn:boudinindicator}
        \E\big[|X_t^{j, H}| \mathds{1}_{\{|X_t^{j, H}|\geq \eps\}} \,\big|\, \mathcal{F}_{t-1}^H \big]
            = H^{-1} \E_{\bar k_t(\hat C_{t-1}^{j, H}, \cdot)}[ |V_t \phi_t(C_t)| \mathds{1}_{\{|V_t \phi_t(C_t)|\geq H\eps\}}],
    \end{equation}
    which uses the definition of $X_t^{j,H}$ and that $(C_t^{j,H},V_t^{j,H})$ are obtained by applying the kernel $\bar{k}_t$ to the $H$-th row of the previous step.
    For any $\alpha > 0$, the function
    \begin{equation}
        c_{t-1} \mapsto \E_{\bar k_t(c_{t-1}, \cdot)}[|V_t \phi_t(C_t)| \mathds{1}_{\{|V_t \phi_t(C_t)|\geq \alpha\}}] \in L^1(\mathcal{C}_{t-1}, \pi_{t-1})
    \end{equation}
    since
    \begin{equation} \label{eqn:boundonrestr}
    	\E_{\tilde \pi_t}[|V_t \phi_t(C_t)| \mathds{1}_{\{|V_t \phi_t(C_t)|\geq \alpha\}}]
    	 	\leq \E_{\tilde \pi_t}[V_t |\phi_t(C_t)|]
    	 	= \E_{\pi_t}[|\phi_t|]
    	 	< \infty.
    \end{equation}
    For $H \geq \alpha/\eps$, we sum \eqref{eqn:boudinindicator} over $j=1, \dots, H$ and get
    \begin{align}
        \sum_{j=1}^{H} \E\big[|X_t^{j, H}| \mathds{1}_{\{|X_t^{j, H}|\geq \eps\}} \, \big| \, \mathcal{F}_{t-1}^H \big]
            &\leq \frac{1}{H} \sum_{j=1}^{H} \E_{\bar k_t(\hat C_{t-1}^{j, H}, \cdot)}[ |V_t \phi_t(C_t)| \mathds{1}_{\{|V_t \phi_t(C_t)|\geq \alpha\}}] \nonumber \\
            &\Pto \E_{\tilde{\pi}_t}[|V_t \phi_t(C_t)| \mathds{1}_{\{|V_t \phi_t(C_t)|\geq \alpha\}}]
            \label{eqn:convofV}
    \end{align}
    as $H \to \infty$, where the convergence follows from \eqref{eqn:convpinduction} and \eqref{eqn:boundonrestr}.
    The right-hand side of \eqref{eqn:convofV} converges to $0$ as $\alpha \to \infty$ by dominated convergence.
    Hence, the left-hand side in \eqref{eqn:convofV} also converges to $0$ in probability which is condition (iii) of Theorem \ref{thm:condwlln}.
    Therefore, we can apply Theorem \ref{thm:condwlln} with our definition of $\mathbf X_t$ and we have shown that for any $\phi_t \in L^1(\mathcal{C}_t, \pi_t)$:
    \begin{equation} \label{eqn:convpagainscondexpt}
        \frac{1}{H} \sum_{j=1}^{H} \left\{ V_t^{j,H} \phi_t(C_t^{j,H}) 
                - \E[V_t^{j,H} \phi_t(C_t^{j,H}) \, | \, \mathcal{F}_{t-1}^H] \right\}
                \Pto 0.
    \end{equation}
    Combining \eqref{eqn:convofcondexp} and \eqref{eqn:convpagainscondexpt} yields
    \begin{equation}
        \frac{1}{H} \sum_{j=1}^{H} V_t^{j, H} \phi_t( C_t^{j, H} ) \Pto \E_{\pi_t}[\phi_t].
    \end{equation}
    Since $c_t \mapsto 1 \in L^1(\mathcal{C}_t, \pi_t)$, the convergence also holds for the ratio estimator
    \begin{equation} \label{eqn:convpratioestimator}
        \sum_{j=1}^{H} \frac{V_t^{j, H}}{\sum_{k=1}^{H} V_t^{k, H}} \phi_t(C_t^{j, H})
        \Pto \E_{\pi_t}[\phi_t].
    \end{equation}
    Since \eqref{eqn:convpratioestimator} does not depend on the normalisation of $V_t^{j,H}$, we have also established the convergence of (\ref{eqn:convpratio}).
    This completes the first part of the proof.
    
    As a second step, we need to show that the SMC estimator (\ref{eqn:convpresampled}) for the resampled triangular array $\hat{ \mathbf C}_t$ converges in probability.
    Let $\hat{\mathcal F}_t^H = \sigma(C^{j, H}_t, V_t^{j, H}; j=1, \dots, H)$ be the $\sigma$-algebra of the $H$-th row before resampling.
    Similar to the previous case, we apply Theorem \ref{thm:condwlln} to the conditionally independent triangular array $\hat{\mathbf X}_t$ given $\{\hat{\mathcal F}_{t}^H\}_H$ with elements
    \begin{equation} 
        \hat X_t^{j, H} := \frac{1}{H} \phi_t(\hat C_t^{j, H}).
    \end{equation}
    The conditional expectation of the resampled particle is
    \begin{equation}
        \E\big[|\hat X_t^{j,H}| \, \big| \, \hat{\mathcal F}_{t}^H\big] 
            = \frac{1}{H} \sum_{j=1}^{H} \frac{V_t^{j, H}}{\sum_{k=1}^{H} V_t^{k, H}} |\phi_t(C_t^{j, H})|
            < \infty \;\; \text{a.s.} ;
    \end{equation}
    this quantity is finite since the weights are positive by assumption and therefore condition (i) of Theorem \ref{thm:condwlln} holds.
    In \eqref{eqn:convpratioestimator}, we have already seen that
    \begin{equation}\label{eqn:convofresampledarray}
        \sum_{j=1}^{H} \E[ \hat X_t^{j, H} \, | \, \hat{\mathcal F}_t^H]
            = \sum_{j=1}^{H} \frac{V_t^{j, H}}{\sum_{k=1}^{H} V_t^{k, H}} \phi_t(C_t^{j, H})
            \Pto \E_{\pi_t}[\phi_t].
    \end{equation}
    As before, the convergence in \eqref{eqn:convofresampledarray} implies boundedness in probability which is condition (ii).
    Furthermore, for any $\alpha>0$ and $H > \alpha/\eps$, we proceed as in (\ref{eqn:convofV})
    \begin{align}
        \sum_{j=1}^{H} \E\big[|\hat X_t^{j, H}| \mathds{1}_{\{|\hat X_t^{j, H}|\geq \eps\}} \, \big| \, \hat{\mathcal F}_{t}^H]
            &\leq \sum_{j=1}^{H} \frac{V_t^{j, H}}{\sum_{k=1}^{H} V_t^{k, H}} \left(|\phi_t| \mathds{1}_{\{|\phi_t|\geq \alpha\}}\right)(C_t^{j, H}) \nonumber \\
            &\Pto \E_{\pi_t}[|\phi_t| \mathds{1}_{\{|\phi_t|\geq \alpha\}}]
    \end{align}
    where the right-hand side converges to $0$ as $\alpha \to \infty$.
    Therefore, condition (iii) holds and we can apply Theorem \ref{thm:condwlln} which yields
    \begin{equation}
        \frac{1}{H} \sum_{j=1}^{H} \left\{ \phi_t(\hat C_t^{j,H}) 
                - \E[\phi_t(\hat C_t^{j,H}) \, | \, \hat{\mathcal F}_{t}^H] \right\}
                \Pto 0.
    \end{equation}
    Combined with \eqref{eqn:convofresampledarray}, we have shown
    \begin{equation}
        \frac{1}{H} \sum_{j=1}^{H} \phi_t(\hat C_t^{j,H}) \Pto \E_{\pi_t}[\phi_t]
    \end{equation}
    which completes the proof.
\end{proof}

\subsubsection{Proof of Theorem \ref{thm:clt}}
Now, we prove the Central Limit Theorem \ref{thm:clt} by induction over the levels $t$ and the different steps of the SMC method.
The strategy is very similar to the proof of Theorem \ref{thm:convinprob}.
To improve readability, we have split the proofs into two lemmas that we combine afterwards to prove the theorem.
\begin{lem}\label{lemma:cltproofpart1}
    Let $t>0$.
    Under the assumptions of Theorem \ref{thm:clt}, assume that we have shown convergence on level $t-1$
    \begin{equation} \label{eqn:cltpreviousstep}
        \frac{1}{\sqrt H} \sum_{j=1}^{H} \left\{ \phi_{t-1}(\hat C_{t-1}^{j, H}) - \E_{\pi_{t-1}}[\phi_{t-1}]\right\}
    		\Dto \mathcal{N}(0, \hat{\Sigma}_{t-1}(\phi_{t-1}))
    \end{equation}
    for $\phi_{t-1} \in A_{t-1}$. 
    Then, for every $\phi \in A_t$ we have
    \begin{align}
        \sqrt{H} \left\{ \sum_{j=1}^{H} \frac{W_t^{j, H}}{\sum_{k=1}^{H} W_t^{k, H}} \phi(C_t^{j, H})
            - \E_{\pi_t}[\phi] \right \}
            &\Dto \mathcal{N}(0, \Sigma_t(\phi))
            \label{eqn:cltproof1}
    \end{align}
    where
    \begin{equation}
    	\Sigma_t(\phi) = \E_{\pi_{t-1}}\left[\var_{\bar k_t}(V_t(\phi(C_t) - \E_{\pi_t}[\phi]))\right] + \hat{\Sigma}_{t-1}\left(\E_{\bar k_t}[V_t (\phi(C_t) - \E_{\pi_t}[\phi])]\right).
    \end{equation}
\end{lem}
\begin{proof}
    We follow the steps of the proof of \cite[Theorem 9.3.7]{cappe2005inference}.
	Without loss of generality, we assume $\E_{\pi_t}[\phi] = 0$.
	Since the estimate in \eqref{eqn:cltproof1} does not depend on the normalisation of the weights, we will prove the result for the normalised weights $V_t^{j, H}$.
	
    Let $\mathcal F_{t-1}^H = \sigma(\hat C_{t-1}^{j, H}, j=1, \dots, H)$ be the $\sigma$-algebra of the previous step.
    We split the level $t$ estimator into two parts
    \begin{align*}
        \frac{1}{H} &\sum_{j=1}^{H} V^{j, H}_t \phi(C_t^{j, H})\\
            &= \frac{1}{H} \sum_{j=1}^{H} \E[V^{j, H}_t \phi(C_t^{j, H}) \, | \, \mathcal{F}_{t-1}^H]
            + \frac{1}{H} \sum_{j=1}^{H} \left\{ V^{j, H}_t \phi(C_t^{j, H}) 
                    - \E[V^{j, H}_t \phi(C_t^{j, H}) \, | \, \mathcal{F}_{t-1}^H] \right\} \\
            &= A_H + B_H. \numberthis
    \end{align*}
    For $A_H$, the convergence
    \begin{equation}\label{eqn:convAH}
        \sqrt{H} A_H = \frac{1}{\sqrt{H}} \sum_{j=1}^{H} \E_{\bar k_t(\hat{C}_{t-1}^{j, H}, \cdot)}[V_t \phi(C_t)]
            \Dto \mathcal{N}\left(0, \hat\Sigma_{t-1}(\E_{\bar k_t}[V_t \phi(C_t)])\right)
    \end{equation}
    follows by Assumption (\ref{eqn:assumpat}) of Theorem \ref{thm:clt} and the convergence in the previous step \eqref{eqn:cltpreviousstep}.
    For $B_H$, we apply Theorem \ref{thm:condclt} to the triangular array $\mathbf X_t$ with elements 
    \begin{equation}
        X_t^{j, H}:= H^{-1/2} V_t^{j, H} \phi(C_t^{j, H}).
    \end{equation}
    By construction, the array $\mathbf X_t$ is conditionally independent given $\{\mathcal F_{t-1}^H\}_H$ and by Assumption (\ref{eqn:assumpat}) of Theorem \ref{thm:clt} we have $c_{t-1} \mapsto \E_{\bar k_t(c_{t-1}, \cdot)}[(V_t \phi(C_t))^2] < \infty$ almost surely.
    This is condition (i).
    Furthermore, we fulfil condition (ii) since
    \begin{align}
        \sum_{j=1}^{H} &\left\{ \E[(X_t^{j, H})^2 \, | \, \mathcal F_{t-1}^H] - (\E[X_t^{j, H} \, | \, \mathcal F_{t-1}^H])^2\right\} \nonumber \\
            &= \frac{1}{H} \sum_{j=1}^{H} \left\{ \E_{\bar k_t(\hat C_{t-1}^{j, H}, \cdot)}[(V_t \phi(C_t))^2]
                - (\E_{\bar k_t(\hat C_{t-1}^{j, H}, \cdot)}[V_t \phi(C_t)])^2 \right\} \nonumber \\
            &\Pto \E_{\pi_{t-1}}[ \E_{\bar k_t(C_{t-1}, \cdot)}[(V_t \phi(C_t))^2] - (\E_{\bar k_t(C_{t-1}, \cdot)}[V_t \phi(C_t)])^2] \nonumber \\
            &= \E_{\pi_{t-1}}[\var_{\bar k_t}(V_t\phi(C_t))]
    \end{align}
    where we used Theorem \ref{thm:convinprob} and $c_{t-1} \mapsto \E_{\bar k_t(c_{t-1}, \cdot)}[(V_t \phi(C_t))^2] \in L^1(\mathcal{C}_{t-1}, \pi_{t-1})$ for the convergence step.
    For condition (iii), we use the same argument as in the proof of Theorem \ref{thm:convinprob}.
    For $\eps > 0$, we have
    \begin{equation} \label{eqn:cltcondexpxsq}
        \E\big[(X_t^{j, H})^2 \mathds{1}_{\{|X_t^{j, H}|\geq \eps\}} \,\big|\, \mathcal{F}_{t-1}^H\big]
            = H^{-1} \E_{\bar k_t(\hat C_{t-1}^{j, H}, \cdot)}[ (V_t \phi_t(C_t))^2 \mathds{1}_{\{|V_t \phi_t(C_t)|\geq \sqrt{H}\eps\}}].
    \end{equation}
    Similar to \eqref{eqn:convofV}, summing \eqref{eqn:cltcondexpxsq} over $j=1, \dots, H$ for any $\alpha > 0$ and $H > (\alpha/\eps)^2$ yields
    \begin{align}
        \sum_{j=1}^{H} \E\big[ (X_t^{j, H})^2 \mathds{1}_{\{|X_t^{j, H}| \geq \eps \}} \,\big|\, \mathcal{F}_{t-1}^H\big]
            &\leq  \frac{1}{H} \sum_{j=1}^{H} \E_{\bar k_t(\hat C_{t-1}^{j, H}, \cdot)}[(V_t \phi(C_t))^2 \mathds{1}_{\{|V_t \phi(C_t)| \geq \alpha\}}] \nonumber \\
            &\Pto \E_{\tilde{\pi}_t}[(V_t \phi(C_t))^2 \mathds{1}_{\{|V_t \phi(C_t)| \geq \alpha\}}]
    \end{align}
    as $H \to \infty$.
    The right-hand side converges to $0$ as $\alpha \to \infty$ and therefore condition (iii) holds.
    In total, we can apply Theorem \ref{thm:condclt} for our definition of $\mathbf X_t$ and get for every $u \in \R$
    \begin{equation}
        \E\left[ \exp(iu \sqrt{H} B_H) \, \middle | \, \mathcal{F}_{t-1}^H \right] 
            \Pto \exp( - \E_{\pi_{t-1}}[\var_{\bar k_t}(V_t\phi(C_t))] u^2/2).
    \end{equation}
    Together with \eqref{eqn:convAH}, this gives us the joint convergence of
    \begin{align}
        \E[\exp&(iu\sqrt{H}(A_H + B_H))] \nonumber \\
            &= \E\left[ \E\left[ \exp(iu \sqrt{H} B_H) \, \middle | \, \mathcal{F}_{t-1}^H \right] \exp(i u \sqrt{H} A_H) \right] \nonumber \\
            &\longrightarrow \exp( -(\E_{\pi_{t-1}}[\var_{\bar k_t}(V_t\phi(C_t))] + \hat{\Sigma}_{t-1}(\E_{\bar k_t}[V_t \phi(C_t)])) u^2/2)
    \end{align}
    which proves the central limit theorem
    \begin{equation} \label{eqn:cltnonratioest}
        \frac{1}{\sqrt{H}} \sum_{j=1}^{H} V^{j, H}_t \phi(C_t^{j, H})
            \Dto \mathcal{N}\left(0, \E_{\pi_{t-1}}[\var_{\bar k_t}(V_t\phi(C_t))] + \hat{\Sigma}_{t-1}(\E_{\bar k_t}[V_t \phi(C_t)])\right).
    \end{equation}
    
     Since $c_{t} \mapsto 1 \in A_t$, we can apply the delta method \cite[Theorem 8.12]{lehmann2006theory} to determine the asymptotic behaviour of the corresponding ratio estimator
     \begin{equation}
            \sqrt{H} \left\{ \sum_{j=1}^{H} \frac{V_t^{j, H}}{\sum_{k=1}^{H} V_t^{k, H}} \phi(C_t^{j, H})
                - \E_{\pi_t}[\phi] \right\}
                \Dto \mathcal{N}\left(0, \Sigma_t(\phi)\right)
     \end{equation}
     with $\Sigma_t$ as defined in Lemma \ref{lemma:cltproofpart1}.
\end{proof}

\begin{lem} \label{lemma:cltproofpart2}
    Under the assumptions of Lemma \ref{lemma:cltproofpart1}, for every $\phi \in A_t$ we have
    \begin{equation}
        \frac{1}{\sqrt H} \sum_{j=1}^{H} \left\{ \phi(\hat C_t^{j, H}) - \E_{\pi_t}[\phi]\right\}
            \Dto \mathcal{N}(0, \hat{\Sigma}_t(\phi))
    \end{equation}
    where
    \begin{equation}
    	\hat \Sigma_t(\phi) = \Sigma_t(\phi) + \var_{\pi_t}(\phi).
    \end{equation}
\end{lem}
\begin{proof}
    We follow the steps of the proof of \cite[Theorem 9.2.14]{cappe2005inference} and Lemma \ref{lemma:cltproofpart1}.
    Without loss of generality, we again assume $\E_{\pi_t}[\phi] = 0$.
    Let $\hat{\mathcal F}_t^H = \sigma(C_t^{j, H}, V_t^{j, H}; j=1, \dots, H)$ be the $\sigma$-algebra of the $H$-th row prior to resampling.
    We can split the SMC estimator into
    \begin{align} \label{eqn:cltlemma2split}
        \frac{1}{H} \sum_{j=1}^{H} \phi(\hat C_t^{j, H})
            &= \frac{1}{H} \sum_{j=1}^{H} \E[\phi(\hat C_t^{j, H}) \, | \, \hat{ \mathcal F}_t^H]
                + \frac{1}{H} \sum_{j=1}^{H} \left\{ \phi(\hat C_t^{j, H}) - \E[\phi(\hat C_t^{j, H}) \, | \, \hat{ \mathcal F}_t^H] \right\} \nonumber \\
            &= \hat A_H + \hat B_H.
    \end{align}
    By the definition of multinomial resampling, we have
    \begin{equation}\label{eqn:cltlemma2defcondexp}
        \E[\phi(\hat C_t^{j, H}) \, | \, \hat{ \mathcal F}_t^H]
            = \sum_{l=1}^{H} \frac{V_t^{l, H}}{\sum_{k=1}^{H} V^{k, H}_t} \phi(C_t^{l, H})
    \end{equation}
    and therefore, we can apply Lemma \ref{lemma:cltproofpart1} to $\hat A_H$
    \begin{equation}  
        \sqrt{H} \hat A_H 
            = \sqrt{H} \sum_{j=1}^{H} \frac{V_t^{j, H}}{\sum_{k=1}^{H} V^{k, H}_t} \phi(C_t^{j, H}) 
            \Dto \mathcal{N}(0, \Sigma_t(\phi)).  
    \end{equation}
    For the convergence of $\hat B_H$, we apply Theorem \ref{thm:condclt} to the triangular array $\hat{\mathbf X}_t$ with elements
    \begin{equation}
        \hat X_t^{j, H} = H^{-1/2} \phi(\hat C_t^{j, H}).
    \end{equation}
    The elements of $\hat{\mathbf X}_t$ are conditionally independent given $\{\hat{\mathcal F}_t^H\}_H$ and their conditional expectation is finite a.s.\ since the weights are positive, see \eqref{eqn:cltlemma2defcondexp}.
    This is condition (i) of Theorem \ref{thm:condclt}.
    To check condition (ii), note that
    \begin{equation}
        \E[(\hat X_t^{j, H})^2 \, | \, \hat{\mathcal F}_t^H ]
            = \sum_{j=1}^{H} \frac{V_t^{j, H}}{\sum_{k=1}^{H} V_t^{k, H}} \phi(C_t^{j, H})^2
            \Pto \E_{\pi_t}[\phi^2]
    \end{equation}
    since $\phi^2 \in L^1(\mathcal{C}_t, \pi_t)$ and we can apply Theorem \ref{thm:convinprob}.
    Hence
    \begin{equation}
        \sum_{j=1}^{H} \left\{ \E[ (\hat X_t^{j, H})^2 \, | \, \hat{\mathcal F}_t^H ]
            - \left( \E [ \hat X_t^{j, H} \, | \, \hat{ \mathcal F}_t^H ] \right)^2 \right\}
            \Pto \var_{\pi_t}[\phi].
    \end{equation}
    Finally, we need to check condition (iii).
    For $\eps > 0, \alpha > 0$, and $H>(\alpha/\eps)^2$, we have
    \begin{align}
        \sum_{j=1}^{H} \E\big[ (\hat X_t^{j, H})^2 \mathds{1}_{\{|\hat X_t^{j, H}| \geq \eps \}} \,\big|\, \hat{\mathcal F}_t^H\big]
            &\leq \sum_{j=1}^{H} \frac{V_t^{j, H}}{\sum_{k=1}^{H} V_t^{k, H}} \left(\phi^2 \mathds{1}_{\{|\phi| \geq \alpha\}}\right)(C_t^{j,H}) \nonumber \\
            &\Pto \E_{\pi_t}[\phi^2 \mathds{1}_{\{|\phi| \geq \alpha\}}]
    \end{align}
    which implies condition (iii) as before.
    Hence we can apply Theorem \ref{thm:condclt} to get
    \begin{equation}
        \E\left[ \exp(iu \sqrt{H} B_H) \, \middle | \, \hat{\mathcal{F}}_t^H \right] 
            \Pto \exp( -\var_{\pi_{t}}[\phi] u^2/2).
    \end{equation}
     Combining the convergence of $A_H$ and $B_H$ as in the proof on Lemma \ref{lemma:cltproofpart1} completes the proof.
\end{proof}

Now, we combine the results of Lemma \ref{lemma:cltproofpart1} and Lemma \ref{lemma:cltproofpart2} to prove Theorem \ref{thm:clt}.
\begin{proof}[Proof of Theorem \ref{thm:clt}]
    For $t=0$, applying the canonical CLT for i.i.d.~random variables yields
    \begin{equation}
        \frac{1}{\sqrt{H}} \sum_{j=1}^{H} \left\{ \phi_0(\hat C_0^{j, H}) - \E_{\pi_0}[\phi_0] \right\}
            \Dto \mathcal{N}(0, \hat \Sigma_0(\phi_0)) \;\; \text{where} \;\; \hat \Sigma_0(\phi_0) = \var_{\pi_0}(\phi_0)
    \end{equation}
    for every $\phi_0 \in A_0$.
    
    For $t>0$, assume that we have shown convergence on level $t-1.$
    Let $\phi_t \in A_t$.
    Lemma \ref{lemma:cltproofpart1} proves the convergence of the ratio estimator
    \begin{align}
        \sqrt{H} \left\{ \sum_{j=1}^{H} \frac{W_t^{j, H}}{\sum_{k=1}^{H} W_t^{k, H}} \phi_t(C_t^{j, H})
            - \E_{\pi_t}[\phi_t] \right \}
            &\Dto \mathcal{N}(0, \Sigma_t(\phi_t))
    \end{align}
    and Lemma \ref{lemma:cltproofpart2} proves the convergence of the resampled array
    \begin{equation}
        \frac{1}{\sqrt H} \sum_{j=1}^{H} \left\{ \phi_t(\hat C_t^{j, H}) - \E_{\pi_t}[\phi_t]\right\}
            \Dto \mathcal{N}(0, \hat{\Sigma}_t(\phi_t)).
    \end{equation}
    This completes the proof of the theorem.
\end{proof}

\subsubsection{Final remarks}
To understand why a.s.\ convergence requires stricter conditions on the quantity of interest $\phi$ and weights $v_t$ compared to the other convergence results, it is useful to compare Theorem \ref{thm:condwlln} and Theorem \ref{thm:condclt} with the strong law of large numbers for triangular arrays in Theorem \ref{thm:slln}.
To apply the strong law, we need to establish the uniform moment bound \eqref{eqn:sllncondition} on the entries in the triangular array, that is for every particle position with non-zero probability.
To ensure this, we had to restrict the a.s.~convergence results to globally bounded quantities of interest and bounded moments of the weights.
In comparison, checking the prerequisites of the theorems in this section only requires conditional independence (which is true by construction), finite conditional expectations (which is easy to check), and limit properties of the conditional expectations.
The latter can be shown using convergence on the previous level.

\bibliographystyle{plain} 
\bibliography{references}

\end{document}